\documentclass[11pt]{article}
\usepackage{amsmath, amssymb, amscd, amsthm, amsfonts}
\usepackage{graphicx}
\usepackage{hyperref}
\usepackage{here}
\usepackage[affil-it]{authblk}
\usepackage{caption}
\usepackage[utf8]{inputenc}
\usepackage{booktabs}
\usepackage{subcaption}
\usepackage[round]{natbib}
\usepackage{bm}

\usepackage{authblk}   
\usepackage{footmisc}  

\title{Tree-Embedded Bayesian Factor Models for Multidimensional Categorical Distributions}

\author[1]{Naoki Awaya$^*$}
\author[2]{Keisuke Sasaki$^*$}
\author[3]{Shonosuke Sugasawa}
\author[4]{Genya Kobayashi}

\affil[1]{School of Political Science and Economics, Waseda University}
\affil[2]{Graduate School of Economics, The University of Tokyo}
\affil[3]{Department of Economics, Keio University}
\affil[4]{School of Commerce, Meiji University}

\date{}

\usepackage{xcolor}

\newtheorem{thm}{Theorem}[section]
\newtheorem{lem}{Lemma}[section]
\newtheorem{cor}{Corollary}[section]
\newtheorem{prop}{Proposition}[section]

\usepackage{amsmath}               
  {
      \theoremstyle{plain}
      \newtheorem{assumption}{Assumption}
  }

\newcommand{\Polya}{P\'{o}lya\ }

\begin{document}
\maketitle

\begin{center}
\vspace{-2em}
\footnotesize{
$^*$These authors contributed equally to this work.
}
\end{center}
\begin{abstract}
    Analyzing data collected from multiple observational units to estimate common and heterogeneous structures through a hierarchical model is a central task in Bayesian inference, and to this end, Bayesian factor models are one of the most widely used tools for this purpose.
    In this paper, we propose a novel Bayesian latent factor model for categorical distributions from grouped data, providing a parsimonious model for describing many observed distributions through lower-dimensional structures.
    Grouped data arise in a wide range of applications in social science, for example, distributions of age composition and income observed across locations.
    In these contexts, standard mixture models can be inefficient because the distributions do not necessarily exhibit clear clustering structures, and the distributions can be more accurately approximated as a combination of lower-dimensional characteristics.
    To analyze distribution-valued data with the Bayesian factor analysis, we adopt a tree-based transformation that embeds distributions into a Euclidean space and construct a Bayesian latent factor model in the transformed space. 
    We develop the hierarchical model by incorporating the infinite factor model, which can adaptively estimate the number of effective factors. 
    In addition, we propose its generalization by incorporating a spatial dependence by introducing a prior based on a simultaneous autoregressive (SAR) model.
    The proposed model provides smooth estimates of multivariate distributional structures, because once a tree-based transformation is applied, both univariate and multivariate distributions are essentially treated as the same Euclidean vectors.
    Through numerical experiments using real population data, we demonstrate that the proposed model outperforms existing parametric and Bayesian nonparametric models in various scenarios involving smooth spatial variations, especially under small sample sizes.
\end{abstract}

\section{Introduction}
Data analysis in a wide range of applications encounters the type of datasets in which each observation itself is a probability distribution rather than a single scalar or vector. 
A motivating example of our research arises in population studies collected in a social science context, where each location is characterized by an empirical distribution describing the demographic composition of stay-population across age and gender categories. 
Similar distribution-valued observations appear in other applications in social science, such as income distributions over predefined income classes across regions or grouped responses in large-scale surveys.
In these contexts, the primary research interest is not in estimating each distribution separately, but in understanding how distributions vary across units and whether such variation can be summarized through a lower-dimensional structure. 
In the population data example, it is natural to consider that the distributions across locations are not independent and to assume that the variability in the distributions can be explained by (possibly a few) characteristics,  such as the ratios of population in young, working, and elderly categories.
Developing a statistical tool that extracts such lower-dimensional structures under the intrinsic constraints of probability distributions is therefore of central importance.

To combine information across multiple observational units while maintaining interpretability, a natural approach is to employ Bayesian factor models \citep{lopes2004bayesian, prado2010time}. 
In this framework, a large number of observed distributions can be explained through a substantially smaller number of latent factors, provided that an appropriate Bayesian factor model for distributions could be introduced. 
Such a model would be well-suited to many situations, including the population data analysis, because it would automatically identify the factors that account for major differences among observations. 
However, applying a factor model directly to distributional data is not straightforward, since standard formulations of factor models assume that observations lie in Euclidean space rather than on a probability simplex.

We bridge the gap between distributional data and standard Bayesian factor models, by transforming distributions through a binary tree partition structure introduced in \citep{wang2026tree}.
As formulated in \cite{wang2026tree} through the ``logistic-tree'' construction, we can establish a bijection between a distribution and a set of conditional probabilities defined on internal nodes of a tree.
Hence, with the logistic transformation to the node probabilities \citep{jara2011class}, we can obtain a Euclidean vector representation of the distribution while preserving the one-to-one relationship with the original distribution. 
We note that the model of \cite{wang2026tree} was designed to estimate mixed effects in compositional data analysis to detect differences in contrasting environments.
In this research, we generalize the idea of the logistic tree transformation by establishing a new direction, namely, proposing a novel general framework of distributional factor models. 
Additionally, as we shall detail, we demonstrate the flexibility of the proposed model by incorporating the spatial information.

An important advantage of adopting a tree-based representation is that it naturally preserves the ordering of categories, allowing nearby values to belong to the same node or to share common ancestors in the tree. 
This property is particularly useful for grouped or continuous data, where adjacent categories are expected to share similar distributional structures. 
Furthermore, the proposed framework avoids parametric assumptions regarding the latent distributional process. 
Introducing parametric assumptions is effective for estimating distributional structures in specific inferential tasks, such as estimating income distributions with uni-modal parametric distributions
\citep{kobayashi2022bayesian, sugasawa2020estimation}, but such a parametric approach does not work under multi-modality or discontinuity of unknown distributions.
In contrast, the tree-based approach allows substantially more flexible estimation, as shown by the studies on the P\'olya tree process models \citep{lavine1992some, hanson2006inference, wong2010optional}.
In the proposed framework, as in the P\'olya tree, we specify flexible generative models through a similar tree-based construction to capture the characteristics of observed distributions without parametric assumptions.

Another advantage of employing a tree-based decomposition is that the framework extends naturally to multidimensional settings, as also demonstrated in analyses based on the P\'olya tree process \citep{wong2010optional, awaya2024hidden}. 
Given a tree-partition structure, multidimensional distributions can be represented through essentially the same recursive construction used for univariate distributions (an example of such a tree is illustrated in Figure~\ref{fig: christmas tree}). 
This property allows us to describe both univariate and multivariate models in a unified framework.
Most importantly, these models can be estimated with the same MCMC algorithm. 
As our proposed model consists of conditionally independent node parameters with the likelihood described in the logistic form, the resulting posterior admits an efficient Gibbs sampler using the \Polya Gamma augmentation \citep{polson2013bayesian}.
This advantage makes the model applicable in a variety of inferential tasks including social science applications, where datasets often involve multiple variables, such as population data with joint age--gender compositions or survey data consisting of multiple questionnaire items.

We note that the existence of an efficient Gibbs sampler is a main advantage of the logistic tree transformation, as opposed to other types of transformation to map distributions into Euclidean space.
For example, several types of transformations have been proposed in the context of compositional data analysis \citep[see, e.g.,][]{aitchison1982statistical, egozcue2003isometric, silverman2017phylogenetic}.
However, as discussed in \cite{wang2026tree}, constructing simple and efficient MCMC samplers under such transformations is less straightforward.

One may argue that mixture-based approaches such as Dirichlet mixture models and Dirichlet process mixtures \citep{muller2015bayesian, ghosal2017fundamentals} can be more advantageous than the factor model. 
However, in many social science applications, distributional data do not exhibit clear clustering structures; rather, distributions often change gradually across locations or other characteristics as shown in the real data analysis provided in Section~\ref{sec: application to the real data}. 
In such cases, mixture models tend to produce an excessively large number of clusters, and it is not possible to obtain a lower-dimensional summary of the data.
In contrast, even when clustering structure is weak or absent, a factor model provides a more parsimonious representation, describing diverse distributions with combinations of a small number of latent factors. 
As an illustration, Figure~\ref{fig: realizations under factor pt} visualizes simulated distributions with four categories generated from the proposed factor model (details are provided in Section~\ref{sec: Tree-Embedded Factor Models for Categorical Distributions}).
In this illustration, we combine only three distributions shown in the first column, which we call ``factor distributions'', and the proposed model can express a variety of distributional structures.

Once each distribution is represented as a Euclidean vector through the logistic transformation, the linear Gaussian likelihood is obtained under the Polya-gamma mixture, so extending the Bayesian hierarchical model becomes straightforward.
The notable features of the proposed framework demonstrating its flexibility are as follows. 
Firstly, we incorporate the multiplicative gamma process of the infinite factor model \citep{bhattacharya2011sparse, durante2017note} into the model, so that we can adaptively estimate lower-dimensional characteristics in the correlation and variation of the observed distributions without fixing the number of factors.
Secondly, in order to estimate the population data using the information of locations effectively, we incorporate a simultaneous autoregressive (SAR) prior on factor loadings to capture spatial dependence across locations. 
We note that these components represent only one possible realization within the proposed framework; alternative Bayesian factor models or additional covariate structures could be incorporated in a similar manner.

Along with the development in the modeling aspect, we provide a theoretical guarantee of the proposed factor model, focusing on the posterior consistency.
One notable finding in our theoretical investigation is that the posterior consistency holds even when the true distributional structure has some cells with exactly zero probability assigned.
The robustness of the logistic tree model to such zero probabilities or zero counts are already discussed in \cite{wang2026tree}, but we first provide a theoretical result to show the robustness. 
We also note that we can establish the posterior consistency even when the model incorporates the spatial correlation, namely, the SAR model as a component in the hierarchical model.

We focus on estimation of grouped data, namely, distributional structures with discrete and ordered categories in this research, as the invention of the novel factor model is mainly motivated by the population data with discrete age and gender categories.
However, we emphasize that continuous distributions are also expected to be analyzed with our proposed model.
For example, in density estimation with the \Polya tree process, it is common to fit tree partitions with finite depth to capture the essential distributional structures \citep[see e.g., ][]{hanson2006inference, soriano2017probabilistic}.
Once a finite tree is fixed, both discrete and continuous distributions can be analyzed with the same model because the estimation in both cases boils down to the posterior inference for the node parameters on the tree.

The remainder of the paper is organized as follows.  
Section~\ref{sec: Tree-Embedded Factor Models for Categorical Distributions} introduces the proposed factor modeling framework for categorical data based on a novel tree-based decomposition of distributions.
We also incorporate spatial correlation to illustrate the flexibility of the framework.
Section~\ref{sec: Application to The Spatiotemporal Population Data} presents an application to stay-population data collected in an urban area of Japan, together with an empirical comparison against alternative models for categorical data.
Section~\ref{sec: discussion} concludes the article.

\section{Tree-Embedded Factor Models for Categorical Distributions}
\label{sec: Tree-Embedded Factor Models for Categorical Distributions}
In this section, we describe the proposed factor model, beginning with a discussion on a tree-based decomposition of grouped data, as it serves as the foundation for the proposed factor model and its spatial extension. 
We then explain how the factor model is constructed on the tree decomposition and how the information of the spatial correlation is incorporated via the SAR model.

\subsection{Logistic-tree decomposition of grouped data}

Let $\Omega$ denote the collection of all finite categories under consideration. 
The elements of $\Omega$ may be indexed by multiple features. 
For example, in population data $\Omega$ is multi-dimensional and indexed by discrete characteristics, such as gender and age. 
We assume that we have a recursive tree-partition structure $T$ that is constructed over the categories in $\Omega$.
In our analysis, we focus particularly on a dyadic tree, namely, a recursive partition in which each node is split into two child nodes, following the common specifications in the Bayesian tree-based models \citep{chipman1998bayesian, chipman2010bart}.
An illustrative example is provided in Figure~\ref{fig: christmas tree}, where the tree is defined for two-dimensional categories. 
In some situations, a natural candidate for a tree structure is available (e.g., phylogenetic trees for microbiome data \citep{wang2026tree}), whereas in many applications an appropriate algorithm is required to construct a tree to allow analysis.
We defer a discussion of the details of tree construction until Section~\ref{sec: tree construction} and in the following discussion assume that the tree $T$ is available and fixed.

Given the tree $T$, we can define a collection of leaf (external) nodes and non-leaf (internal) nodes, denoted by $\mathcal{L}(T)$ and $\mathcal{N}(T)$, respectively. 
To denote the nodes of the tree, we adopt the recursive expression following the literature on the \Polya tree process \citep{soriano2017probabilistic, awaya2024hidden}. 
A generic node in the tree is denoted by $A$, and if $A \in \mathcal{N}(T)$, it is partitioned into two child nodes, denoted by $A_l$ and $A_r$. 
We let the number of internal nodes be denoted by $N$, resulting in $N+1$ categories in total.

We can define a discrete measure (distribution) on the categories in $\Omega$, and this distribution is denoted by $P$.
Throughout the discussion, we focus on the collection of discrete distributions with positive masses, namely, 
\[
\mathcal{P}_{N+1}
=
\left\{
    (p_1,\dots,p_{N+1}):
    p_1 \in (0,1),\ \dots\ ,p_{N+1} \in (0,1),\ 
    \sum^{N+1}_{l=1}p_l = 1
\right\}.
\]
A key property of distributions defined along a tree partition is that 
the distribution is uniquely characterized by conditional probabilities defined on the internal nodes:
\[
\theta(A) := P(A_l \mid A), \qquad A \in \mathcal{N}(T),
\]
 as shown in the literature on the P\'olya tree process (e.g., \cite{lavine1992some, hanson2006inference}).
Hence, posterior inference for the distribution reduces to estimating the parameters $\theta(A)$. 
Figure~\ref{fig: visualization of tree decomp} illustrates a tree with two layers and three internal nodes ($A$, $A_l$, and $A_r$), where a discrete distribution with four masses ($p_1,p_2,p_3,p_4$) is parameterized by $\theta(A)$, $\theta(A_l)$, and $\theta(A_r)$.
This parameterization is more advantageous than directly estimating the masses $p_1,\dots,p_{N+1}$ because it introduces independence (or, conditional independence in a hierarchical model) of the node parameters, which adds flexibility to the hierarchical modeling and makes posterior sampling efficient. 

\begin{figure}[htb]
    \centering
    \begin{tabular}{cc}
        \includegraphics[width = 0.5\textwidth]{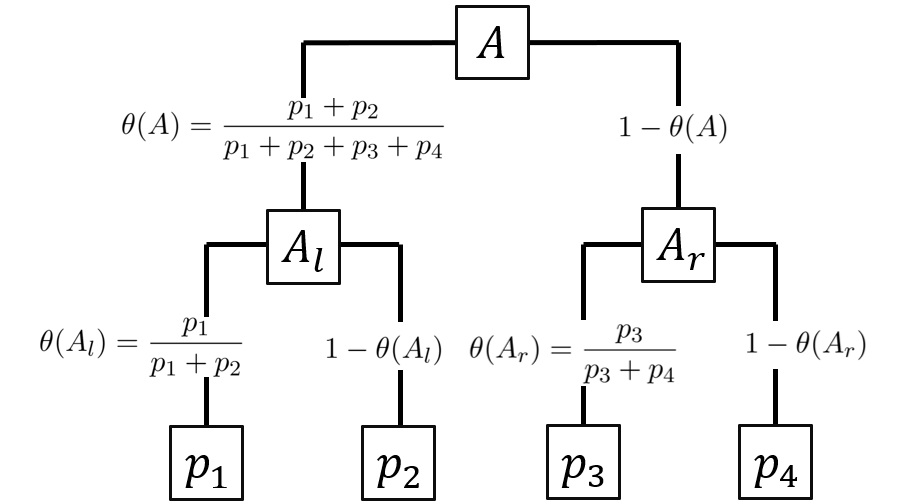}
    \end{tabular}
    \caption{Illustration of a two-layer tree with three internal-node parameters.}
    \label{fig: visualization of tree decomp}
\end{figure}

The support of each node parameter $\theta(A)$ is limited to the finite interval $(0,1)$, so, for example, many P\'olya tree models adopt beta priors for conjugacy \citep{ma2011coupling, soriano2017probabilistic, awaya2024hidden}.  
In this work, however, in order to facilitate the introduction of a factor model, we follow the logistic-tree normal formulation \citep{jara2011class, wang2026tree}, commonly used for real-valued vectors, and introduce logit-transformed parameters
\[
\psi(A) := \log\!\left(
\frac{\theta(A)}{1-\theta(A)}
\right).
\]
This transformation defines a bijection between the original parameters $\{\theta(A)\}_{A \in \mathcal{N}(T)} \subset [0,1]^N$ and the transformed parameters $\{\psi(A)\}_{A \in \mathcal{N}(T)} \subset \mathbb{R}^N$.
Equivalently, this transformation defines a tree-based embedding $\mathcal{E}_T : \mathcal{P}_{N+1} \to \mathbb{R}^N$ by
\[
    \mathcal{E}_T(P)
    =
    \{\psi(A)\}_{A \in \mathcal{N}(T)},
\]
which is bijective and connects a multinomial distribution and a real-valued vector in the unconstrained space $\mathbb{R}^N$.
The transformed representation is particularly convenient for incorporating Gaussian-based modeling strategies, because it allows us to build a hierarchical model in which
$\mathcal{E}_T(P)$ corresponding to an observed real vector in the original model specification.

Under this embedding, \cite{jara2011class} formulated a hierarchical model with a Gaussian process prior and \cite{wang2026tree} demonstrated richer latent structures through graphical-lasso-type priors and mixture models. 
As our goal is to summarize the information of many observed distributions as in the population data analysis, we establish a new approach of utilizing the tree embedding by introducing a novel factor model for distributional observations.
\\

\noindent
{\it Remark:}
The tree-based transformation should be differentiated from an approach of directly transforming the observed proportions into a Euclidean vector, which is common in compositional data analyses 
\citep{aitchison1982statistical, egozcue2003isometric}.
In the proposed model, we apply the transformation to the unknown multinomial distribution $P_i$, from which observable proportions are generated, hence the generative model is estimated. 
This modeling approach is particularly beneficial in Bayesian estimation, because we can successfully quantify uncertainty in  estimation, which is caused by the difference in total counts.
For example, we can consider two cases in which we observe exactly the same proportions and different total counts.
We can make a distinction between these two cases by specifying the generative models because the difference in the total counts results in differences in several aspects, such as posterior variances and credible intervals.

\subsection{Factor models for discrete distributions}
We introduce a novel factor model that represents a collection of distributions through a smaller number of latent factor distributions. 
Suppose that we observe $M$ distributions defined on the category set $\Omega$. 
For $i = 1,\dots,M$, let $n_i(B)$ denote the number of observations falling in a subset $B \subset \Omega$. 
We assume that the counts $n_i$ are generated from an unknown distribution $P_i$. 
Using the tree-based embedding $\mathcal{E}_T$, we obtain a real vector representation
\[
\mathcal{E}_T(P_i)
=
\{\psi_i(A)\}_{A \in \mathcal{N}(T)}.
\]
For notational convenience, we denote the parameter vector $\{\psi_i(A)\}_{A \in \mathcal{N}(T)}$ by $\bm{\psi}_i$.
The central idea of factor analysis is to explain high-dimensional observations, here $M$ distributions, through a lower-dimensional structure characterized by $K$ ($K < M$) latent factors. 
Since the unknown distributions are now represented by a Euclidean vector, we can adopt the standard formulation of Bayesian factor models \citep{lopes2004bayesian, prado2010time}. 
To clearly describe the model architecture, we first consider the case where the number of factors $K$ is fixed. 

We model each embedded vector $\bm{\psi}_i$ as
\begin{align}
\bm{\psi}_i 
=
\lambda_{i1}\bm{\eta}_1 + \cdots + \lambda_{iK}\bm{\eta}_K,
\label{eq: factor model expression 1 without mu}
\end{align}
where $\lambda_{i1},\dots,\lambda_{iK} \in \mathbb{R}$ are scalar factor loadings and $\bm{\eta}_1,\dots,\bm{\eta}_K$ are $N$-dimensional factors.
As in the standard factor models, the information of $M$ latent distributions $P_1,\dots,P_M$ is summarized in $\bm{\eta}_1,\dots,\bm{\eta}_K$.
Because the tree-logit transformation is bijective, each $\bm{\eta}_k$ implicitly corresponds to a set of conditional probabilities on the tree $T$, and hence to a discrete distribution on $\Omega$. 
We call this distribution representing the information of a factor ``a factor distribution.''
Under this formulation, the model combines the $K$ factor distributions
\[
\mathcal{E}_T^{-1}(\bm{\eta}_1),\dots,
\mathcal{E}_T^{-1}(\bm{\eta}_K)
\]
within the Euclidean embedding space, instead of defining their mixture distribution. 
As a generalized model specification, we can introduce a mean vector 
$\bm{\mu} = \{\mu(A)\}_{A \in \mathcal{N}(T)}$ and write the model as
\begin{align}
\bm{\psi}_i 
=
\bm{\mu}
+
\lambda_{i1}\bm{\eta}_1 + \cdots + \lambda_{iK}\bm{\eta}_K.
\label{eq: factor model expression 1 with mu}
\end{align}
The mean $\bm{\mu}$ can be fixed to reflect prior structural information, for example, based on the number of categories associated with each node, or treated as an unknown parameter. 
Our preliminary experiments suggest that estimating $\bm{\mu}$ jointly with the factor component introduces weak identifiability between $\bm{\mu}$ and the latent factors and makes the posterior sampling unstable.
In practice, therefore, we fix the mean $\mu(A)$ by computing the average of the logit-transformed empirical conditional probabilities.

An immediate but important property of the proposed model is that the intrinsic dimension of the induced space of distributions equals the number of factors $K$, allowing a parsimonious representation of many distributions.
To illustrate this property, consider a simplex space with four categories equipped with the symmetric partition tree $T$ and three factor distributions
\[
d_1 = 
(1/6,1/6,1/3,1/3),\  
d_2 = 
(1/2,1/6,1/6,1/6),\ 
d_3 
= (1/6,1/6,1/6,1/2).
\]
Figure~\ref{fig: realizations under factor pt} illustrates the three factor distributions and a sample of distributions that can be generated from the model
\[
\mathcal{E}^{-1}_T
(
\lambda_{i1}\mathcal{E}_T(d_1)
+
\lambda_{i2}\mathcal{E}_T(d_2)
+
\lambda_{i3}\mathcal{E}_T(d_3)
),
\]
under several values of the loadings $(\lambda_1, \lambda_2, \lambda_3)$. 
It is seen that the logistic-tree factor model can represent a rich family of distributions on the simplex.
In theory, any distribution on $3$-simplex can be approximated through a set of suitable factor coefficients with the three linearly independent basis vectors. 
This expressive power is obtained because the loadings $(\lambda_1, \lambda_2, \lambda_3)$ are unconstrained real values rather than restricted to $(0,1)$. 
For instance, for the component $d_1$, the coefficient $\lambda_{1}$ may be negative, allowing the model to assign both high and low probabilities to the first and second categories.
Additionally, by increasing $\lambda_1$, we can even emphasize the characteristics of $d_1$.

\begin{figure}[htb]
    \centering
        \includegraphics[width = 0.8\textwidth]{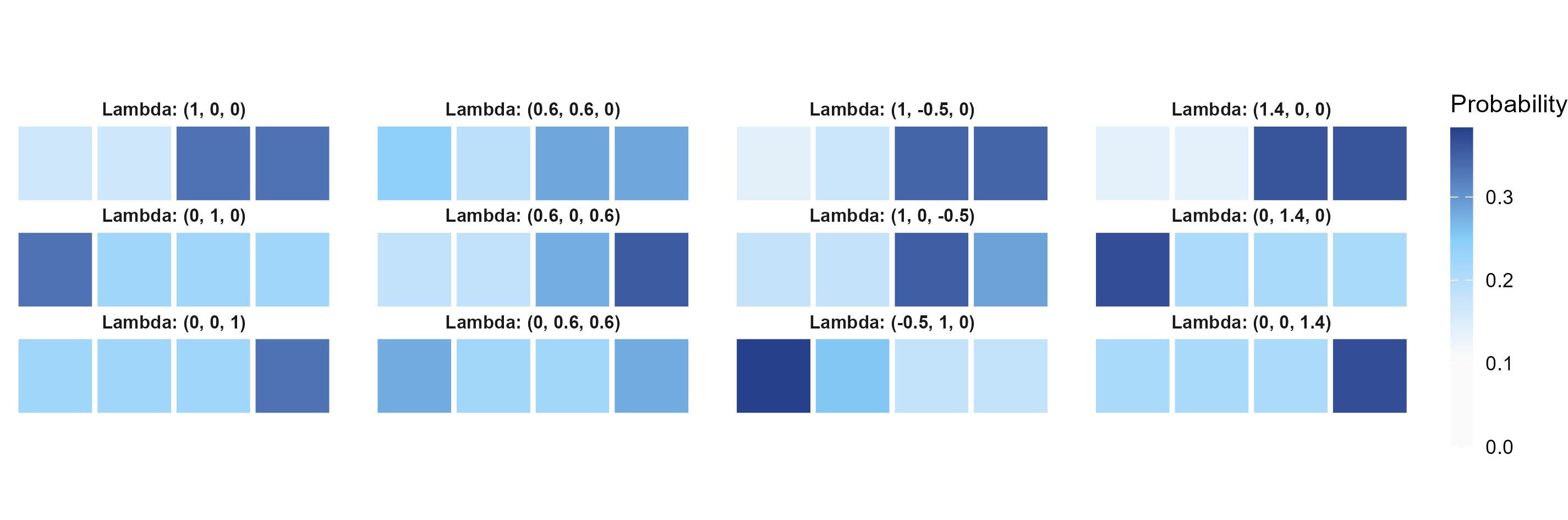}
    \caption{An illustration of realized distributions that can be generated in the proposed factor model under 12 different combinations of factor loadings $(\lambda_1, \lambda_2, \lambda_3)$.
    The first column corresponds to the three factor distributions.
    }
    \label{fig: realizations under factor pt}
\end{figure}

\subsection{Infinite factor model and its generalization with spatial correlations}
To describe the prior specification, it is convenient to express the factor model in a nodewise form for each internal node $A \in \mathcal{N}(T)$. 
An equivalent representation of Eq~\eqref{eq: factor model expression 1 with mu} is given by
\begin{align}
\boldsymbol{\psi}(A)
=
\mu(A)\,\iota_M
+
\boldsymbol{\Lambda}\,\boldsymbol{\eta}(A),
\label{eq: factor model expression (nodewise)}
\end{align}
where $\iota_M$ denotes the $M\times 1$ vector of ones, and
\[
\boldsymbol{\psi}(A)
=
\begin{bmatrix}
\psi_1(A)\\
\vdots\\
\psi_M(A)
\end{bmatrix},
\qquad
\boldsymbol{\eta}(A)
=
\begin{bmatrix}
\eta_1(A)\\
\vdots\\
\eta_K(A)
\end{bmatrix},
\]
and $\boldsymbol{\Lambda}$ is the $M\times K$ loading matrix
\[
\boldsymbol{\Lambda}
=
\begin{bmatrix}
\lambda_{11} & \lambda_{12} & \cdots & \lambda_{1K}\\
\lambda_{21} & \lambda_{22} & \cdots & \lambda_{2K}\\
\vdots & \vdots & \ddots & \vdots\\
\lambda_{M1} & \lambda_{M2} & \cdots & \lambda_{MK}
\end{bmatrix}.
\]
The factors $\eta_k(A)$ ($k=1,\dots,K$ and $A \in \mathcal{N}(T)$) independently follow the standard Gaussian distribution.

Careful prior specification for $\boldsymbol{\Lambda}$ is essential, since factors and loadings are not identifiable under rotations and switching signs. 
A common strategy is to impose structural constraints on the loading matrix, for example by setting upper-triangular entries $\lambda_{ik}$ ($i<k$) to zero \citep{geweke1996measuring, aguilar2000bayesian}. 
Such approaches, however, require a pre-specified ordering of observations and our fixing a number of factors, both of which may substantially affect posterior inference. 
Instead, we adopt the infinite factor model of \cite{bhattacharya2011sparse} to estimate the lower-dimensional structure adaptively without such pre-specifications.

Moreover, in many applications, additional structural information is available and should be incorporated through the prior. 
For example, in population data, the spatial locations of observations are known, and it is desirable to borrow strength across nearby regions. 
To achieve this, we introduce a simultaneous autoregressive (SAR) prior for the loading vectors based on an adjacency matrix $W$, whose rows are normalized to sum to one. 
Although we focus on the SAR specification, the proposed framework is compatible with alternative spatial priors, such as the CAR model considered in \cite{wakayama2024spatiotemporal}.

Within the infinite factor framework \citep{bhattacharya2011sparse}, the prior on loadings is designed so that their scales tend to decrease with the factor index, and the essential structure of the variability in the observations is summarized by the first few factors.
While the model has a countably infinite number of factors in theory, we choose an initial truncation level $K$ moderately large (e.g., $K=10$) in order to provide a good approximation to the correlation structure. 
Based on this framework, for the $k$th column of the loading matrix $\bm{\Lambda}_{\cdot k}$, we introduce a combination of the sparsity inducing prior from \cite{bhattacharya2011sparse} and the SAR model. 
Let $\rho_k \in (-1,1)$ be the correlation parameter and $W$ is the row-normalized adjacency matrix.
Then, we introduce the prior that is induced by the following model:
\[
(I_M - \rho_k W)\boldsymbol{\Lambda}_{\cdot k}
\sim
N\!\left(\mathbf{0}, (\tau_k \Phi_k)^{-1}\right),
\]
where $\Phi_k = \mathrm{diag}(\phi_{1,k},\dots,\phi_{M,k})$ introduces local shrinkage with
\[
\phi_{i,k} \sim \mathrm{Gamma}\!\left(\frac{\nu}{2},\frac{\nu}{2}\right), 
\qquad i=1,\dots,M,
\]
and $\tau_k$ is the global precision parameter following the multiplicative gamma process \citep{bhattacharya2011sparse}:
\[
\tau_k = \prod_{l=1}^{k}\delta_l,\qquad
\delta_1\sim \mathrm{Gamma}(a_1,1),\qquad
\delta_l\sim \mathrm{Gamma}(a_2,1)\ (l\ge2).
\]
In this prior, the difference $\bm{\Lambda}_{\cdot k} - \rho_k W\bm{\Lambda}_{\cdot k}$ represents residual variation that is not explained by neighboring locations, to which sparsity-inducing shrinkage is applied. 
The precision matrix of the loading vector $\bm{\Lambda}_{\cdot k}$ under the Gaussian distribution is
\[
    \tau_k
    (I_M - \rho_k W^T)
    \Phi_k
    (I_M - \rho_k W). 
\]
It implies that variance decreases as the factor index increases due to the multiplicative gamma prior on $\tau_k$. 
For the correlation parameter $\rho_k$, we adopt a truncated normal prior $\mathrm{TN}_{(-1+c,1-c)}(m_\rho,s_\rho^2)$, where $c$ (e.g., $0.001$) is a small constant introduced to avoid the unstable behavior of the MCMC.

Following \cite{durante2017note} and \cite{de2021bayesian}, we set $(a_1,a_2)=(2.1,3.1)$ and $\nu=3$. 
For $\rho_k$, we specify a weakly informative prior by setting $(m_\rho,s_\rho^2)$ to $(0, 10)$ accordingly.

\subsection{Posterior sampling}
\label{sec: Posterior sampling}
Deriving the posterior sampling algorithm is straightforward because, as discussed in \cite{wang2026tree}, the logistic-tree formulation allows the introduction of P\'olya--Gamma (PG) augmentation \citep{polson2013bayesian}. 
We briefly outline the PG augmentation for the proposed model here, while detailed sampling steps are provided in Appendix~\ref{appendix: MCMC algorithm}.

For the location $i$ ($i=1,\dots,M$) and internal node $A \in \mathcal{N}(T)$, the likelihood of observing $n_i(A_l)$ counts in $A_l$ out of $n_i(A)$ total observations follows the binomial distribution and can be written as
\begin{align*}
p(n_i(A_l)\mid n_i(A),\psi_i(A))
&=
\left(
\frac{\exp(\psi_i(A))}
{1+\exp(\psi_i(A))}
\right)^{n_i(A_l)}
\left(
\frac{1}
{1+\exp(\psi_i(A))}
\right)^{n_i(A_r)}\\
&=
\frac{\exp\!\left(\psi_i(A)\right)^{n_i(A_l)}}
{\left\{1+\exp\!\left(\psi_i(A)\right)\right\}^{n_i(A)}}.
\end{align*}
Following \cite{polson2013bayesian}, this likelihood admits a P\'olya--Gamma mixture representation:
\begin{align*}
&p(n_i(A_l)\mid n_i(A),\psi_i(A))\\
&\propto 
\exp\!\big(\kappa_i(A)\psi_i(A)\big)
\int_0^\infty
\exp\!\left\{-\frac{\omega_i(A)\psi_i(A)^2}{2}\right\}
p(\omega_i(A)\mid n_i(A),0)\,
d\omega_i(A),
\end{align*}
where $\kappa_i(A)=n_i(A_l)-n_i(A)/2$ and $\omega_i(A)\sim \mathrm{PG}(n_i(A),0)$.

Based on the augmentation, the model admits a pseudo-linear Gaussian representation. 
Specifically we define the $M\times1$ vector $\bm{\kappa}(A)$ and the diagonal matrix $\bm{\Omega}(A)$ by
\[
    \bm{\kappa}(A)
    =
    \left[
        \begin{array}{c}
            \kappa_1(A)\\
            \kappa_2(A) \\
            \vdots \\
            \kappa_M(A)
        \end{array}
    \right],\ 
    \bm{\Omega}(A)
    =
    \left[
        \begin{array}{cccc}
            \omega_1(A) & 0 & \cdots & 0\\
            0 & \omega_2(A) &  & \vdots  \\
            \vdots & & \ddots & 0 \\
            0 & \cdots & 0 & \omega_M(A).
        \end{array}
    \right].
\]
Then, given $\bm{\Omega}(A)$, the likelihood for the parameters related to the internal node $A$ can be written as the pseudo-regression model
\begin{align}
\bm{\Omega}(A)^{-1}\bm{\kappa}(A)
=
\mu(A)\iota_M
+
\boldsymbol{\Lambda}\boldsymbol{\eta}(A)
+
\boldsymbol{\epsilon}(A),
\label{eq: linear factor model representation}
\end{align}
where $\iota_M=(1,\dots,1)^T$ denotes the $M \times 1$ vector of $1$, and $\boldsymbol{\epsilon}(A)\sim N(\mathbf{0},\bm{\Omega}(A)^{-1})$ is an auxiliary Gaussian noise term. 
Hence, conditional on the latent variables $\bm{\Omega}(A)$, the likelihood reduces to a Gaussian linear model with pseudo-observation $\bm{\Omega}(A)^{-1}\bm{\kappa}(A)$. 
The resulting MCMC sampling algorithm follows from those available for standard Gaussian models and is detailed in Appendix~\ref{appendix: MCMC algorithm}.

\subsection{Post-processing for the factor loadings}
\label{subsec: Post-processing for the factor loadings}
The infinite factor model employed in this work has strong expressive power for capturing the correlation structure among the $M$ observed distributions.  
We can show the posterior consistency under the model, as is to be discussed in Section~\ref{sec: theoretical properties}. 
For interpretability, however, it is effective to determine the number of factors that adequately explain dependence structure for better interpretability of the estimation result. 
In addition, the prior on the loading matrix does not impose any constraints that resolve the non-identifiability problem arising from sign switching and rotational invariance \citep[see, e.g.,][]{papastamoulis2022identifiability}, so the summarizing the MCMC sample without modification would be misleading. 
To give a solution to these problems, we follow the post-processing strategy of \cite{de2021bayesian}.

\subsubsection{Selecting the optimal number of factors}
As implied by the model formulations in Eq~\eqref{eq: factor model expression (nodewise)} and \eqref{eq: linear factor model representation}, and noting that the factors $\bm{\eta}(A)$ follow standard Gaussian priors, the dominant component governing the correlation among the $M$ distributions is $\boldsymbol{\Lambda}\boldsymbol{\Lambda}^\top$, a symmetric positive semi-definite matrix of rank $K$. 
When $\boldsymbol{\Lambda}$ is given, the effective dimensionality of the factor representation can be assessed through the ordered eigenvalues, analogously to the analysis based on scree plots. 
Following \cite{de2021bayesian}, we determine the effective number of factors $K^*$ as follows. 
First, at each MCMC iteration we compute the ordered eigenvalues of $\boldsymbol{\Lambda}\boldsymbol{\Lambda}^\top$, based on which their posterior means are estimated. 
Second, based on the cumulative proportion of explained variation, we select $K^*$ as the number of eigenvalues whose contribution exceeds a pre-specified threshold (e.g., 5\%).

\subsubsection{Addressing the identifiability problem}
After determining the effective number of factors $K^*$ ($K^* < K$), we again run the MCMC to obtain the MCMC draws denoted by $\{\boldsymbol{\Lambda}^{(r)}\}_{r=1}^R$ and post-process the sample to address the identification problem with the rotation-sign-permutation algorithm in \cite{papastamoulis2022identifiability}.
Their algorithm first applies a varimax rotation to each MCMC draw to obtain a simple loading structure, and then aligns the rotated loading matrices across iterations by solving a signed-permutation problem relative to an iteratively updated reference loading matrix.
This algorithm outputs a sequences of orthonormal matrices 
$\{\boldsymbol{Q}^{(r)}\}^R_{r=1}$ 
, so we can compute the identification-corrected loading matrix $\boldsymbol{\Lambda}^{(r)} \boldsymbol{Q}^{(r)}$ and the corresponding factor $\left(\boldsymbol{Q}^{(r)}\right)^\top \boldsymbol{\eta}(A)$. 

\noindent
{\it Remark: }
Our preliminary experiment using simulated data shows that this post-processing step tends to make the errors from the true distributions larger, and sometimes their difference is noticeable.
This phenomenon, however, is expected since the post-processing step makes the factor substantially more restrictive by dropping $K - K^*$ factors.
Hence, the post-processing should be understood as a procedure to summarize the estimation result with a small number of factors and thus to make it more interpretable.

\subsection{On the tree construction algorithms}
\label{sec: tree construction}
In this section, we discuss the tree building algorithm, which is crucial to enable the proposed logistic-tree method,
In the literature on tree-based models for estimating distributions, such as Bayesian algorithms to sample from the posterior of trees  \citep[e.g.,][]{wong2010optional, awaya2024hidden}, a prior distribution is introduced for the tree structures, and the computation methods explore the resulting posteriors.
Such computation algorithms are feasible when we can evaluate the marginal likelihood for the trees, integrating out the variables defined on the trees
However, such a framework is infeasible for the proposed model since the model is equipped with the infinite factor model, which has a complex dependence structure in the latent variables. 
Hence, we adopt an alternative strategy and fix the tree structure prior to implementing the factor analysis.

Since our primary goal in the factor analysis is to estimate the structure in the variability among the observed distributions, one plausible strategy is to seek a tree where the difference among the distributions is most highlighted.
To this end, we adopt the principal balances algorithm \citep{pawlowsky2011principal}, which constructs tree structures based on isometric log-ratio (ilr) coordinates \citep[See][for details]{egozcue2003isometric}. 
Within the ilr framework, an orthonormal basis is defined through a binary tree partition, and each distribution is characterized the corresponding coordinates known as a ``balance.''
For example, when each internal node $A$ split into $A_l$ and $A_r$, a given distribution yields a balance defined as,
\[
\sqrt{\frac{|A_l||A_r|}{|A_l|+|A_r|}}
\log \frac{g(A_l)}{g(A_r)},
\]
where $|A_l|$ and $|A_r|$ denote the numbers of categories in the respective nodes, and $g(A_l)$ and $g(A_r)$ represent the geometric means of probabilities within each subset. 
The variability of balances across observed distributions quantifies how strongly a given partition captures differences among distributions, conceptually analogous to variance explained by rotated coordinates in the principal component analysis.

We employ the maximum explained variance hierarchical balances (MV) algorithm of \cite{pawlowsky2011principal}, which follows a coarse-to-fine greedy strategy. 
When splitting a node $A$, candidate partitions are evaluated by computing the variance of the resulting balances, and the split maximizing this variance is selected. 
While the original MV algorithm evaluates only a subset of candidate splits, the relatively small number of categories in our applications allows us to exhaustively evaluate all possible partitions.

As an illustration, Figure~\ref{fig: christmas tree} displays the tree constructed from the population data observed at 19:00 on 12/24/2019. 
The two-dimensional categories are first divided into two age groups, indicating that variation among locations is primarily driven by the ratio of elderly people. 
The same tendency is observed for the subsequent splits. 
According to the computed balance, a large variation tends to be explained by the difference in the age compositions.

\begin{figure}[htb]
    \centering
    \includegraphics[width=0.8\linewidth]{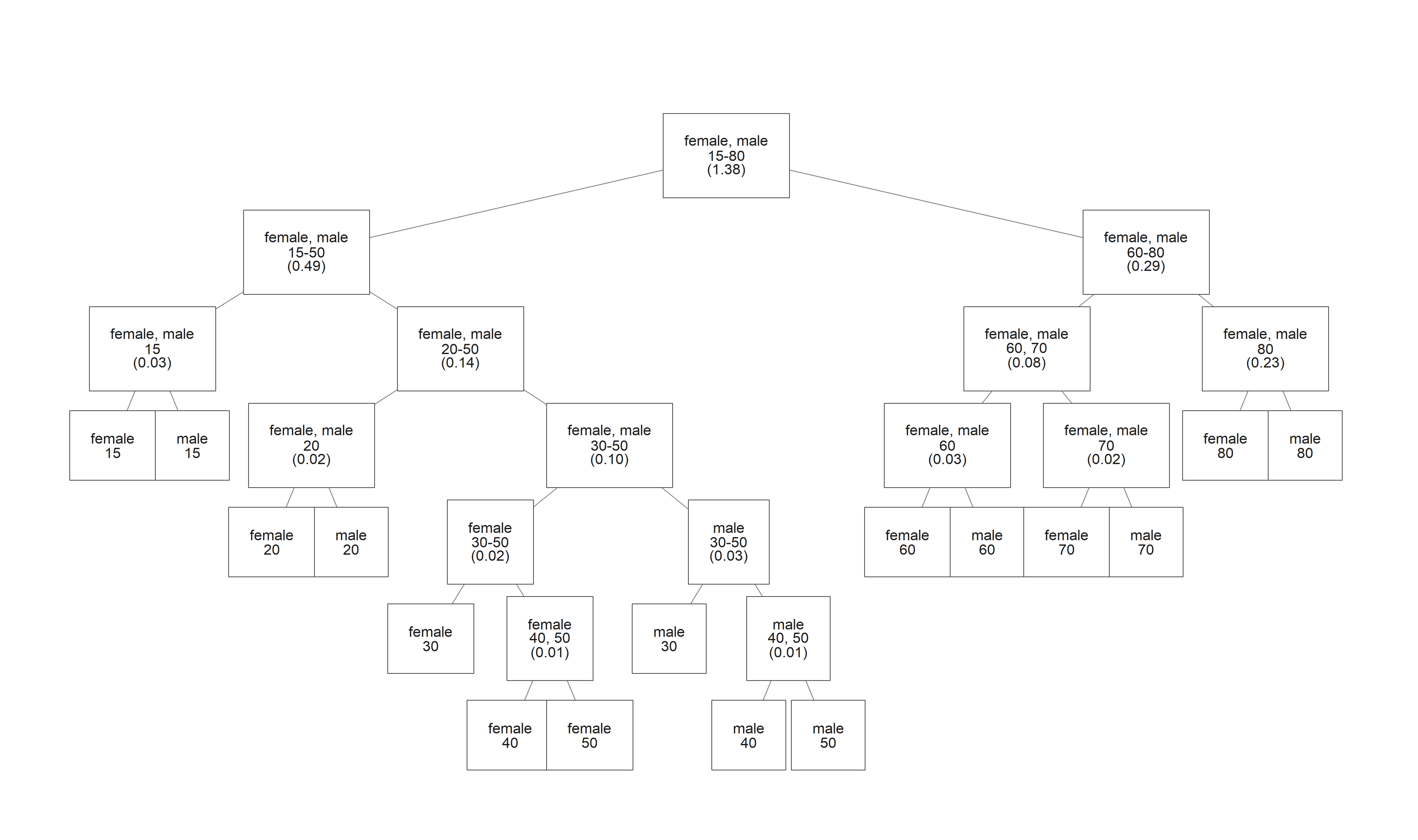}
    \caption{Visualization of the tree partition obtained using the MV algorithm for the population data observed at 19:00 on 12/24/2019.
    For the internal nodes, the variances of the balances calculated for all locations are also shown in the parentheses.}
    \label{fig: christmas tree}
\end{figure}

It should be noted that the MV algorithm is discussed just as the default choice, and other types of data-driven tree algorithms, such as the median-based $k$d tree \citep{walther2023beta}, can also be applicable. 
We can also construct a tree without referring to data by, for example, building a symmetric tree that assigns the same number of categories into every pair of child nodes, so one might be concerned with the possible sensitivity to the selection of trees.
To address this possible concern, we carry out the sensitivity analysis evaluating five tree constructing algorithms, including those discussed above, in Appendix~\ref{appendix: additional numerical results}.
The analysis shows that all tree algorithms considered exhibit similar performance, implying little sensitivity to the choice of trees, and this finding is consistent with the sensitivity analysis carried out in \cite{wang2026tree}.

\section{Theoretical properties of the proposed factor model}
\label{sec: theoretical properties}

In this section, we establish the asymptotic properties of the proposed methods, particularly focusing on its consistency.
In particular, we generalize the results on the posterior consistency for the infinite factor models for Euclidean vector observations \citep{bhattacharya2011sparse}, to the proposed hierarchical model with spatially-correlated latent factors for high dimensional categorical distributions. 
Although the results are presented formally in the subsequent sections, it is worth summarizing the main findings.
First, while \cite{bhattacharya2011sparse} established posterior consistency for infinite factor models under conditionally independent factor loadings, we extend their result to the setting with spatially correlated loadings induced by the SAR-type prior.
Second, the consistency result is established for categorical distributions. 
In particular, our theory covers boundary cases where some true category probabilities are exactly zero, thereby providing formal support for the ability of logistic-tree models to address the zero-cell issue discussed by \cite{wang2026tree} without an explicit zero-inflation component.

\subsection{Assumptions on the true data-generating process}
In order to discuss the asymptotic property of the proposed model, we need to introduce the true categorical distribution for each location $i$ ($i=1,\dots,M$), denoted by $P^{(0)}_i = \{P^{(0)}_{i,c}\}^{N+1}_{c=1}$.
Unlike the discussion on the model specification (Section~\ref{sec: Tree-Embedded Factor Models for Categorical Distributions}), the true probability mass is allowed to be exactly zero.
We assume that $P^{(0)}_i$ is an element of the following collection of simplexes;
\[
\mathcal{P}'_{N+1}
=\left\{
    (p_1,\dots,p_{N+1}):
    p_1 \in [0,1],\ \dots\ ,p_{N+1} \in [0,1],\ 
    \sum^{N+1}_{l=1}p_l = 1
\right\}.
\]
It should be noted that, the elements of this simplex space are allowed to have cells with exactly zero probability.
Also, for the sake of a simplified discussion, we assume that we have $n$ i.i.d observations generated from the product measure
\begin{align}
    P^{(0)} = \bigotimes_{i=1}^M P^{(0)}_i.
    \label{eq: true dist}
\end{align}
In other words, we consider a setting in which we observe the same total counts for the $M$ locations.
Under this assumption, the KL divergence of the true distribution $P^{(0)}$ from that generated from the proposed Bayesian model admits the simple expression 
\[
\sum_{i=1}^{M} D_{KL}\!\left(P_i^{(0)}\,\|\,P_i\right),
\]
where $P_i$ is the categorical distribution for the $i$th distribution generated in the model.
This simplification is effective to establish the posterior consistency because we use the Schwartz's theorem \citep{schwartztheorem, ghosal2017fundamentals} and show that the proposed model assigns a positive probability to any KL-neighborhood of the true distribution.
See the proof given in Appendix~\ref{appendix:proofs}. 
We note that this setting is only for theoretical discussion. 
The proposed model also works with the simple MCMC even for distributions with different total counts, as already discussed in Section~\ref{sec: Posterior sampling}.

We introduce the necessary conditions to derive the theoretical results.

\begin{assumption}
The $n$ observations are independently and identically distributed under the distribution defined in Eq~\eqref{eq: true dist}.
Also, the hyperparameters in the proposed factor model satisfy $\nu > 2$ and $a_2 > 2$.
\end{assumption}
The assumptions for the hyperparameters, which are also introduced in \cite{bhattacharya2011sparse}, are crucial to ensure that the scales of the factor loadings degenerate toward zero with $k$ with an appropriate rate.

\subsection{Properties of the distributional factor model and posterior consistency}
\label{subsec: Properties of the distributional factor model and posterior consistency}
Recall that, in our proposed model, the $i$th distribution ($i=1,\dots,M$) is represented by a Euclidean vector $\{\psi_i(A)\}_{A \in \mathcal{N}(T)}$, and each element $\psi_i(A)$ is, in theory, defined as a combination of countably infinite factors.
Hence the proposed model involves infinite summations, but whether they are finite or infinite with positive probability is not clear.
The following property is analogous to that shown in \cite{bhattacharya2011sparse} but more specific to our distributional factor model with the spatial correlation structure, showing that the infinite sums appearing in the model are finite almost surely.

\begin{prop}
For the proposed factor model, the following holds:
    \begin{enumerate}
    \item[(i)] $P\!\left(\max\limits_{1 \le i \le M}\sum_{k=1}^\infty \lambda_{ik}^2 < \infty\right) = 1$.
    \item[(ii)] $  P\!\left(\max\limits_{1 \le i \le M,\ A \in \mathcal{N}(T)}
            \left|
                \psi_{i}(A)
            \right|
            < \infty
        \right)
        = 1.$
  \end{enumerate}
\end{prop}
The proofs of this property and the following results are all provided in Appendix~\ref{appendix:proofs}.

In practice, we truncate the infinite structure up to, say, $K$ factors, and as a result the factor vector $\bm{\psi}(A)$ is a vector with $K$ elements and the loading matrix $\bm{\Lambda}$ is a $M \times K$ matrix.
The corresponding finite approximation to $\psi_i(A)$ is denoted by $\psi_{i,K}(A)$.
The next proposition implies that we can surely improve the accuracy of the finite approximation by making $K$ large and the difference can be arbitrarily small.

\begin{prop}
\label{prop: approximation}
For any $\xi > 0$, $\bm{\psi}_K(A)$ converges to $\bm{\psi}(A)$ in probability as $K \to \infty$ in the sense that
\[
    \lim_{K \to \infty} P\!\left(\max\limits_{1 \le i \le M,\ A \in \mathcal{N}(T)}|\psi_{i}(A) - \psi_{i,K}(A)| > \xi\right) = 0.
\]
\end{prop}
This property on the approximation result is also effective in showing the posterior consistency.
In the proof, we use the Schwartz's theorem \citep{schwartztheorem, ghosal2017fundamentals}, and to this end, we show that the prior assigns positive probability to an arbitrary small KL neighborhood of the true distribution $P^{(0)}$.
Thus, intuitively speaking, the totality of $\psi_{i}(A)$ needs to become arbitrary close to the ``true $\psi$'' with positive probability, and in this proof, Proposition~\ref{prop: approximation} is effective because it allows us to instead consider the closeness between the true structure and the finite version $\psi_{i,K}(A)$, which is substantially more convenient to analyze.

By utilizing the properties of the proposed factor model introduced above, we can establish the posterior consistency.
\begin{thm}
  \label{thm: consistency}
  Under the infinite factor prior, the KL support condition holds: for any $\epsilon > 0$,
  \[
    P\!\left(\sum_{i=1}^{M} D_{KL}\!\left(P_i^{(0)}\,\|\,P_i\right) < \epsilon\right) > 0.
  \]
  Consequently, by Schwartz's theorem, the posterior distribution is weakly consistent.
\end{thm}

Although Theorem~\ref{thm: consistency} is stated under the SAR prior, the same posterior consistency conclusion is expected to extend to other spatial dependence structures, such as CAR priors or Gaussian process priors, whenever the induced prior on the loading process is sufficiently well behaved to control the infinite-factor representation and sufficiently flexible to approximate the true distribution. In that case, the prior continues to assign positive mass to Kullback--Leibler neighborhoods of the truth, and the consistency argument goes through in the same way. Concrete sufficient conditions and the corresponding proof strategy are discussed in Appendix~\ref{appendix:proof_spatial}.

\section{Numerical Experiments}
\label{sec: Numerical Experiments}

In this section, we evaluate the proposed distributional factor model through simulation experiments and an empirical analysis of population data.
The simulation experiments are designed to examine how the proposed model performs under various types of distributional structures, such as smooth heterogeneity and multidimensional categorical structures.

We evaluate the proposed factor model, using the MV tree builder described in Section~\ref{sec: tree construction} and a simple splitting algorithm (``Midpoint'') that always splits the node at the middle point and generates child nodes with equal number of categories after splitting.
(The details of the tree building algorithms, including those used in the sensitivity analysis, are described in Appendix~\ref{sec: Sensitivity Analysis on tree building algorithms}.)
We show a comparison with two benchmark methods: a Dirichlet process mixture model (DPM) and an independent Dirichlet model (ID).
The ID benchmark fits a separate Dirichlet-multinomial model at each location and therefore does not borrow information across locations. 
This section provides a comparison between the proposed factor model with the SAR-type correlation and the benchmark models.
One may be concerned with the fairness of the experiment since the benchmark models (especially, DPM) do not incorporate the spatial information.
Thus, we conducted another type of comparison using the proposed model with no spatial correlation with $\rho_k = 0$. 
Appendix~\ref{appendix: additional numerical results} provides the result, which is essentially the same as the spatial case. 

For the MCMC sampling, we use 5000 draws simulated after the burn-in period of size 1000. 
The maximum number of factors $K$ is set to $10$.
Additional results for fixed spatial correlation and alternative tree builders are reported in the Appendix~\ref{sec: Sensitivity Analysis on tree building algorithms}.

In both experiments, to evaluate the performance of the competing methods, we simulate 50 different data sets for each scenario and calculate the Kullback--Leibler (KL) divergence and the Hellinger distance.

In all simulation scenarios considered in the experiments, we simulate $M=200$ locations, by independently sampling each location from the uniform distribution defined on the two-dimensional set $(-1,1]^2$, and a subset of the other locations within the radius $0.15$ is considered as adjacent neighbors to define the adjacency matrix $W$.
The coordinates for the $i$-th location are denoted by $(x_i, y_i)$.
The number of observations in each location is determined by the tuning parameter $n$ and uniformly distributed between $0.5 n$ and $2n$.
The setting of $n$ changes according to scenarios.

Given the locations, we define various types of spatial correlation structures defined on one-dimensional and three-dimensional spaces.
We emphasize that all scenarios are introduced to check if the proposed model works in general situations, and in  {\it none} of the scenarios we assume that the data-generating process is correctly specified by the factor model.
In the following description, the true distribution defined for the $i$-th location ($i=1,\dots,M$) is denoted by $P^{(0)}_i$, as in Section~\ref{sec: theoretical properties}.
The detailed descriptions of the generating models are provided in Appendix~\ref{appendix: Details of the numerical experiments}.

\subsection{One-dimensional scenarios}
\label{subsec: One-dimensional scenarios}
The one-dimensional scenarios are defined on the sample space consisting of ordered eight categories obtained by discretizing the finite continuous space $(0, \infty)$ at the knots $20, 30, \dots, 80$ to simulate artificial age distributions.
Each scenario specifies latent continuous distributions and discretizes them to define the true location distribution $P^{(0)}_i$.
We consider the following two scenarios for the latent continuous distributions.

\begin{description}

\item[1. ``Two-parameters'']
This scenario is motivated by \cite{sugasawa2020estimation}, and the latent distribution at the $i$-th location is specified by the log-normal distribution
\[
\mathrm{LogNormal}(\mu_i,\sigma_i^2),
\]
where
\[
\mu_i = 0.1 + x_i^2 + y_i^2,
\qquad
\log \sigma_i^2 = 0.2x_i + 0.2y_i.
\]

\item[2: ``Mixture'']
The latent continuous distribution at the \(i\)-th location is a mixture of three log-normal distributions
\[
\sum_{k=1}^3
w_{ik}
f_{\mathrm{LN}}(a;\mu_k,\sigma_k^2),
\]
where the mixture weight $w_{ik}$ depends on the coordinates $(x_i, y_i)$.
The details such as the functional form of $w_{ik}$ and the parameters $\mu_k$ and $\sigma^2_k$ are provided in Appendix~\ref{appendix: Details of the numerical experiments}.
\end{description}

As an illustration, Figure~\ref{fig: example of the mixture scenario} visualizes one simulated data set for the ``mixture'' scenario by showing the realized locations and the true distribution $P^{(0)}_i$ at the selected locations.
From the figure, we can see that nearby locations have similar distributions, whereas distant groups show different distributional shapes.

\begin{figure}[htb]
    \centering
        \includegraphics[width = 0.85\textwidth]{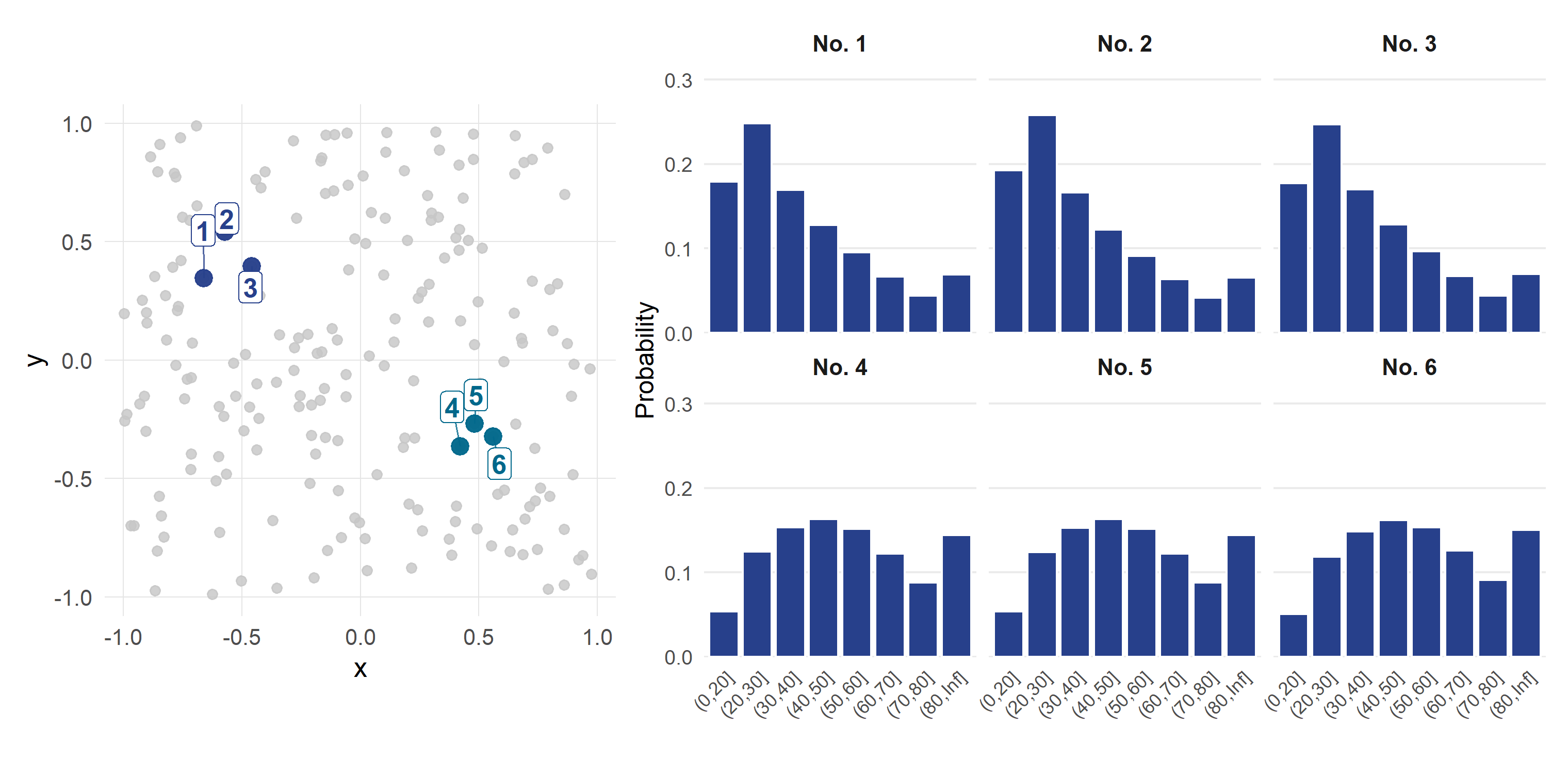}
    \caption{An illustration of the one-dimensional ``mixture'' simulation scenario.
The left panel shows 200 generated locations and six selected locations, grouped as nearby triples.
The right panels display the true probability vectors $P^{(0)}_i$ at the six selected locations.
    }
    \label{fig: example of the mixture scenario}
\end{figure}

We carried out the experiment under two sample size settings ($n=200$ and $n=1000$), and the results are summarized in Table~\ref{tab:simulation-1d-default-small}.
The results show that the proposed factor model outperforms both benchmark models, DPM and ID, in both one-dimensional scenarios and under both sample size settings.
The MV and Midpoint tree constructions give almost identical performance, suggesting that the estimation accuracy is not highly sensitive to this particular choice of tree construction in these one-dimensional settings.

The comparison with DPM clarifies the difference between mixture-based and factor-based representations.
In these one-dimensional scenarios, the true distributions vary smoothly over the location space rather than forming clearly separated clusters.
In such settings, representing the variation by a small number of latent factors is more natural than assigning each distribution to a latent cluster.
This explains why the proposed factor model can estimate the smoothly varying distributions more accurately than DPM.

Moreover, the comparison with ID highlights the importance of borrowing information across locations.
Although ID is flexible as a model for each categorical distribution, it estimates each location-specific distribution independently.
Consequently, it cannot make use of a common low-dimensional structure or spatial smoothness shared across locations.
The deterioration of ID under the reduced sample size indicates that such information sharing is especially important when the local count is limited.

\begin{table}[htbp]
\centering
\scriptsize
\caption{Simulation results for the one-dimensional scenarios under the original and reduced sample sizes. Entries show mean (SD) over 50 replicates. Reported values are $10^2 \times$ KL divergence and $10^2 \times$ Hellinger distance; lower values are better for both metrics.}
\resizebox{\textwidth}{!}{%
\begin{tabular}{llllcc}
\toprule
 & & & & \multicolumn{2}{c}{Metric} \\
\cmidrule(lr){5-6}
Scenario & Size & Model & Method & $10^2 \times$ KL div. & $10^2 \times$ Hellinger \\
\midrule
Two-param. & $n=1000$ & Factor (proposed) & MV & \textbf{0.061(0.006)} & \textbf{1.075(0.055)} \\
 &  &  & Midpoint & \textbf{0.061(0.006)} & 1.076(0.055) \\
 &  & Benchmark & DPM & 0.163(0.015) & 1.810(0.081) \\
 &  &  & ID & 0.280(0.015) & 2.477(0.065) \\
 & $n=200$ & Factor (proposed) & MV & 0.217(0.025) & 2.013(0.109) \\
 &  &  & Midpoint & \textbf{0.216(0.025)} & \textbf{2.007(0.107)} \\
 &  & Benchmark & DPM & 0.513(0.050) & 3.186(0.153) \\
 &  &  & ID & 1.353(0.071) & 5.425(0.131) \\
\midrule
Mixture & $n=1000$ & Factor (proposed) & MV & \textbf{0.044(0.005)} & 0.931(0.056) \\
 &  &  & Midpoint & \textbf{0.044(0.005)} & \textbf{0.924(0.059)} \\
 &  & Benchmark & DPM & 0.081(0.011) & 1.182(0.080) \\
 &  &  & ID & 0.327(0.018) & 2.697(0.073) \\
 & $n=200$ & Factor (proposed) & MV & 0.163(0.023) & 1.736(0.103) \\
 &  &  & Midpoint & \textbf{0.161(0.024)} & \textbf{1.731(0.104)} \\
 &  & Benchmark & DPM & 0.258(0.035) & 2.093(0.139) \\
 &  &  & ID & 1.614(0.088) & 5.955(0.155) \\
\bottomrule
\end{tabular}%
}
\label{tab:simulation-1d-default-small}
\end{table}

\subsection{Three-dimensional scenarios}
\label{subsec: Three-dimensional scenarios}

The three-dimensional scenarios are defined on the product sample space $\mathcal X
=
\mathcal A \times \mathcal B \times \mathcal C$, where each set consists of ordered categories as
\[
\mathcal A=\{1,\ldots,8\},\qquad
\mathcal B=\{1,2\},\qquad
\mathcal C=\{1,2,3\}.
\]
For the \(i\)-th location, the true joint probability $P^{(0)}_i$ is denoted by
\[
P^{(0)}_i(a,b,c),
\qquad
(a,b,c)\in \mathcal A\times\mathcal B\times\mathcal C.
\]
The first dimension \(\mathcal A\) is interpreted as the same ordered eight-category partition used in the one-dimensional scenarios, while the other two dimensions are treated as additional categorical coordinates.
We consider the following three scenarios.

\begin{description}

\item[1. ``3D axis-separated'']
In this scenario, the joint distribution factorizes as
\[
P^{(0)}_i(a,b,c)
=
p_i^{(1)}(a)p_i^{(2)}(b)p_i^{(3)}(c).
\]
The first margin \(p_i^{(1)}\) is generated from a nearly common ordered distribution, while the second and third margins vary smoothly along different spatial coordinates.
Specifically, \(p_i^{(2)}\) changes mainly with \(x_i\), and \(p_i^{(3)}\) changes mainly with \(y_i\).
This scenario is introduced to represent a relatively simple three-dimensional structure in which the main spatial variation is separated across the categorical dimensions.

\item[2. ``3D latent'']
In this scenario, the joint distribution again takes the product form
\[
P^{(0)}_i(a,b,c)
=
p_i^{(1)}(a)p_i^{(2)}(b)p_i^{(3)}(c),
\]
but the three marginal distributions are driven by a common latent spatial variable.
Let \(U_i\) denote a latent variable generated from a smooth Gaussian process over the two-dimensional location space.
Conditional on \(U_i\), each margin is generated through a logistic or softmax model.
Thus, this scenario creates dependence among the three categorical dimensions through a shared latent spatial factor, even though the margins are conditionally independent given \(U_i\).

\item[3. ``3D cluster'']
In this scenario, the spatial domain is divided into nine rectangular regions, and each region is assigned a cluster-specific joint distribution on
\(\mathcal A\times\mathcal B\times\mathcal C\).
For location \(i\), let \(c(i)\) denote its spatial cluster membership.
Then the true distribution is given by
\[
P^{(0)}_i(a,b,c)
=
\pi_{c(m)}(a,b,c),
\]
where each \(\pi_r\) ($r=1,\dots,9$) is a cluster-specific probability vector.
This scenario is intentionally favorable to mixture-based methods, since the true data-generating mechanism has a spatially clustered structure.

\end{description}

We carried out the experiment under the two settings for sample size ($n=1000$ and $n=5000$), and the results are summarized in Table~\ref{tab:simulation-3d-default-small}.
The overall pattern in the ``3D axis-separated'' and ``3D latent'' scenarios is consistent with the one-dimensional results.
Although these scenarios are more challenging because the probability vector is defined on a product space over 48 categories, the proposed factor model performs well when the true distributions vary smoothly over the spatial domain.
In particular, it clearly outperforms DPM and ID under the reduced sample size, and it gives the best or nearly best performance under the larger sample size.
The MV and Midpoint tree constructions again give very similar results.

These results demonstrate that the advantage of the factor representation is not limited to one-dimensional ordered categories.
Even in multidimensional categorical spaces, the proposed model can exploit common low-dimensional variation across locations once the distributions are embedded through the tree transformation.

The result for the ``3D cluster'' scenario should be interpreted separately.
In this scenario, the data-generating process is explicitly cluster-based, and therefore DPM is expected to have an advantage.
Indeed, DPM gives the best performance in this scenario.
This result does not contradict the usefulness of the proposed factor model; rather, it clarifies the type of structure for which each model is well suited.
When the true distributions form clear clusters, a mixture model is naturally appropriate.
When distributions vary continuously over locations without clear cluster boundaries, the factor model provides a more suitable representation.

Further, the fact that the proposed factor model outperforms the mixture model is not the most essential result.
As discussed above, a main advantage of using a factor model is that we can summarize the information of the data using lower-dimensional characteristics so that the result becomes interpretable.
We demonstrate this interpretability using the real population data in Section~\ref{sec: Application to The Spatiotemporal Population Data}.

\begin{table}[htbp]
\centering
\scriptsize
\caption{Simulation results for the three-dimensional scenarios under the original and reduced sample sizes. Entries show mean (SD) over 50 replicates. Reported values are $10^2 \times$ KL divergence and $10^2 \times$ Hellinger distance; lower values are better for both metrics.}
\resizebox{\textwidth}{!}{%
\begin{tabular}{llllcc}
\toprule
 & & & & \multicolumn{2}{c}{Metric} \\
\cmidrule(lr){5-6}
Scenario & Size & Model & Method & $10^2 \times$ KL div. & $10^2 \times$ Hellinger \\
\midrule
3D axis-sep. & $n=5000$ & Factor (proposed) & MV & \textbf{0.038(0.002)} & 0.937(0.030) \\
 &  &  & Midpoint & \textbf{0.038(0.002)} & \textbf{0.936(0.028)} \\
 &  & Benchmark & DPM & 0.462(0.035) & 3.201(0.118) \\
 &  &  & ID & 0.437(0.012) & 3.216(0.042) \\
 & $n=1000$ & Factor (proposed) & MV & 0.155(0.010) & 1.882(0.057) \\
 &  &  & Midpoint & \textbf{0.153(0.009)} & \textbf{1.875(0.054)} \\
 &  & Benchmark & DPM & 1.071(0.118) & 4.857(0.261) \\
 &  &  & ID & 2.154(0.074) & 7.072(0.110) \\
\midrule
3D latent & $n=5000$ & Factor (proposed) & MV & 0.114(0.009) & 1.642(0.071) \\
 &  &  & Midpoint & \textbf{0.113(0.010)} & \textbf{1.640(0.078)} \\
 &  & Benchmark & DPM & 0.554(0.028) & 3.621(0.090) \\
 &  &  & ID & 0.430(0.016) & 3.194(0.056) \\
 & $n=1000$ & Factor (proposed) & MV & \textbf{0.534(0.055)} & \textbf{3.596(0.205)} \\
 &  &  & Midpoint & 0.535(0.055) & 3.599(0.206) \\
 &  & Benchmark & DPM & 1.059(0.135) & 4.918(0.288) \\
 &  &  & ID & 2.084(0.070) & 7.015(0.105) \\
\midrule
3D cluster & $n=5000$ & Factor (proposed) & MV & 0.111(0.011) & 1.907(0.113) \\
 &  &  & Midpoint & 0.111(0.011) & 1.907(0.109) \\
 &  & Benchmark & DPM & \textbf{0.017(0.001)} & \textbf{0.643(0.025)} \\
 &  &  & ID & 0.417(0.015) & 3.221(0.055) \\
 & $n=1000$ & Factor (proposed) & MV & 0.713(0.053) & 4.991(0.222) \\
 &  &  & Midpoint & 0.713(0.057) & 4.993(0.230) \\
 &  & Benchmark & DPM & \textbf{0.082(0.005)} & \textbf{1.434(0.048)} \\
 &  &  & ID & 1.967(0.063) & 7.083(0.103) \\
\bottomrule
\end{tabular}%
}
\label{tab:simulation-3d-default-small}
\end{table}

\section{Application to the Spatial Population data}
\label{sec: application to the real data}
In this section, we first provide a quantitative evaluation of the proposed factor model and the benchmark models with the real spatial population data to check the efficacy of the models in a realistic situation.
After that, we discuss the implications obtained with the proposed model to clarify that, with the factor model, we can characterize the observed more than 400 distributions with a lower dimensional structure, and thus the result is highly interpretable.

\subsection{Evaluation of the models}
\label{subsec: Evaluation with real Population data}
In this empirical study, we analyze the stay-population data provided by NTT Docomo, one of the largest mobile carriers in Japan with more than 90 million accounts. 
Based on their operational data, the number of mobile terminals in each base station area is counted. 
Then, the population of each area is extrapolated with high accuracy using the penetration rate of NTT Docomo (see \cite{terada2013population, oyabu2013evaluating}). 
We consider the 500 meter mesh data, including $M=452$ observation areas (meshes), on the five special wards of Tokyo (Figure~\ref{fig:tokyomap}). 
A sample of the observed distributions is also illustrated in Figure~\ref{fig:tokyomap}.
The same data set has been used in different Bayesian contexts of spatial or spatiotemporal analysis \citep[see, for example,][]{wakayama2024spatiotemporal,wakayama_jrss}. 

\begin{figure}[htb]
    \centering
    \begin{tabular}{cc}
    \includegraphics[width=0.35\linewidth]{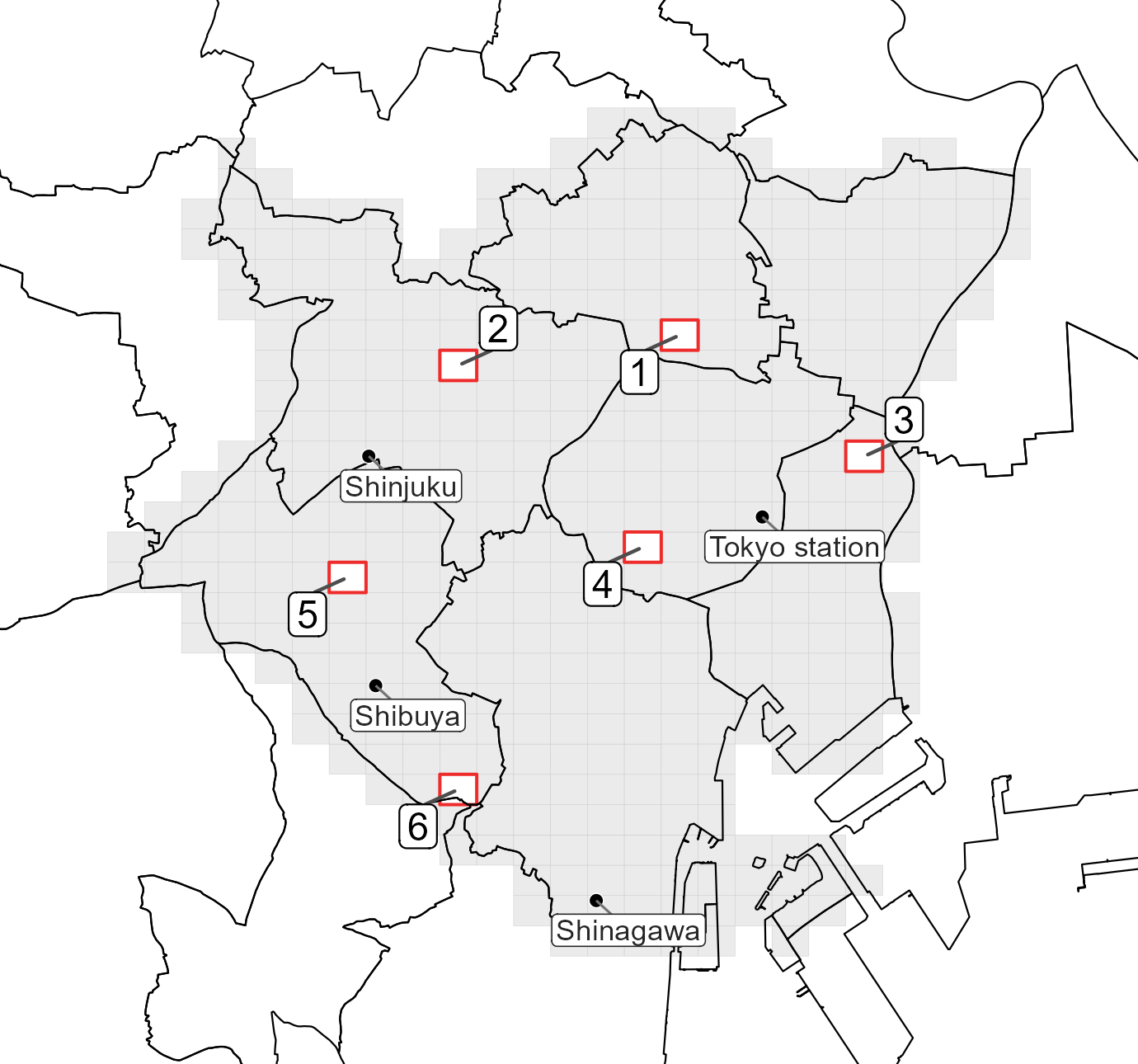}   
    &
    \includegraphics[width=0.6\linewidth]{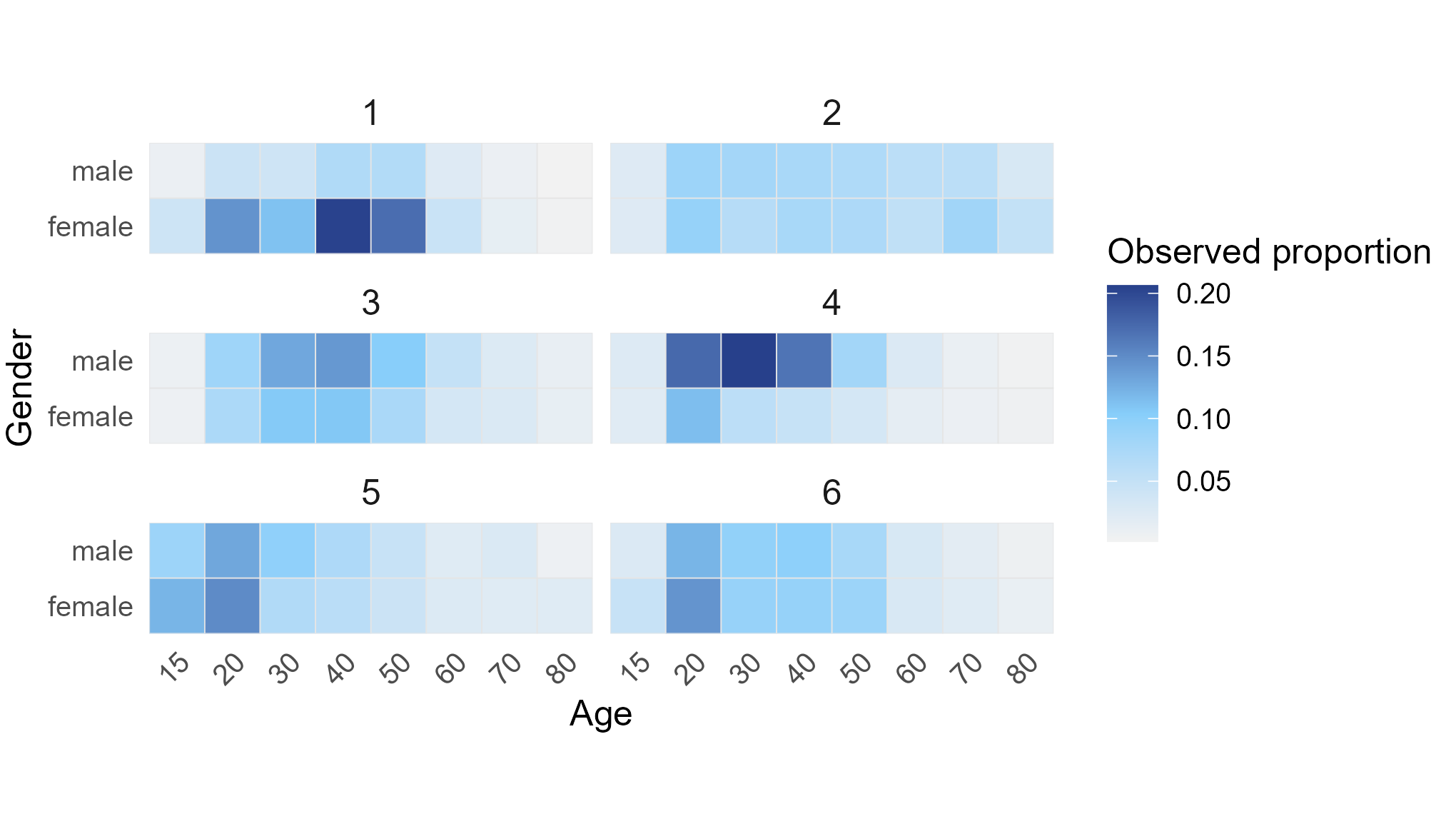}
    \end{tabular}
    \caption{The left plot shows a map of the study area where the spatiotemporal population data were collected (colored gray), alongside the locations of four representative districts or landmarks discussed in the analysis, and six specific areas whose observed distributions are visualized in the right plot. 
    }
    \label{fig:tokyomap}
\end{figure}

In contrast to the previous studies above, which only considered populations of gender-age-total, in this analysis, we consider the distributions of population by gender and age. 
Specifically, the dataset provides the two gender categories (Male and Female) and eight age groups: 15--19, 20--29, 30--39, 40--49, 50--59, 60--69, 70--79, and 80 years or older. 
Consequently, for each area and time point, the observation is represented as a frequency vector with $N+1=16$ categories. 
In this analysis, we extract nine specific time points from February, July, and December 2019 with various temporal characteristics, such as morning/night and weekday/holiday, as detailed in Table~\ref{tab:ppl-real-scale100-with-independent-dirichlet}, and independently analyze the data for the selected time points. 
Each area has a relatively large sample size; for example, for the first selected time (National Holiday), the average total count is approximately 6000. 
In order to evaluate the robustness of the methods under different amounts of information, we use the original data sets and also modified data in which all counts are reduced to one tenth by multiplying them with $0.1$ and then rounding them to integers.

Since the information on the true data-generating process is unavailable, the primary evaluation metric for comparing the three models is the Posterior Predictive Loss (PPL) defined in \citep{gelfand1998PPLdef}, where a smaller value indicates better trade-off between model fit and model complexity. 
Adapting their original formulation to the specific structure of the population data, we define the PPL for a given model $\mathcal{M}$ as
\begin{equation*}
    \mathrm{PPL}(\mathcal{M}) = \frac{1}{M}\sum_{i=1}^{M} \sum_{j=1}^{N+1} V_{ij}^{\mathcal{M}} + \frac{1}{M+1}\sum_{i=1}^{M} \sum_{j=1}^{N+1} (y_{ij}^{obs} - E_{ij}^{\mathcal{M}})^2,
\end{equation*}
where $y_{ij}^{obs}$ represents the observed count for the $j$-th category in the $i$-th area, and $E_{ij}^{\mathcal{M}}$ and $V_{ij}^{\mathcal{M}}$ denote the mean and variance of the posterior predictive distribution under model $\mathcal{M}$, respectively.

Table~\ref{tab:ppl-real-scale100-with-independent-dirichlet} provides a comparison of the PPLs for the proposed model and the benchmark models for the original data, Table~\ref{tab:ppl-real-scale010-with-independent-dirichlet} provides that for the artificially reduced data.
The results show a clear difference between the original and reduced data settings.
For the original data, the ID benchmark gives the smallest PPL values across all nine time points.
This result is reasonable because each location has a relatively large count, so estimating each location-specific categorical distribution independently is already accurate for predictive purposes.
In contrast, under the reduced data, the proposed Factor PT models result in the smallest PPL values for all time points.
This indicates that, when local counts are limited, borrowing information across locations through the common factor structure and the spatial prior is effective.

The comparison with DPM also supports the same interpretation as in the simulation experiments.
In both data settings, DPM gives larger PPL values than the best Factor PT model, and it is also worse than ID for most targets on the original scale.
These results suggest that the population distributions over the selected Tokyo area are not well summarized by a small number of discrete clusters.
This finding is consistent with the fact that we observed that the estimated number of clusters tends to be large for the original data and can be greater than 100 in several scenarios, as detailed in Appendix~\ref{appendix: additional numerical results}.
Hence, the variation in the observed distributions is better viewed as gradual spatial heterogeneity, for which the factor representation is more appropriate.

The role of the spatial component is most visible in the reduced data.
For the original data, the difference in the PPL between the spatial model where $\rho$ is estimated and non-spatial models with $\rho=0$ is relatively small, while the non-spatial model tends to result in slightly smaller PPL. 
This implies that the large local counts already contain enough information to accurately estimate each distribution.
After reducing the counts, however, the spatial models consistently exhibit improvement over the non-spatial models. 
Thus, when the local sample size is small, the SAR component would be most useful in stabilizing estimation through spatially borrowing information. 

\begin{table}[htbp]
\centering
\scriptsize
\caption{Posterior predictive loss (PPL) on the nine real-data targets for the original scale. Factor PT is shown with the MV and midpoint tree builders, each with either estimated $\rho$ or fixed $\rho=0$. 
}
\resizebox{\textwidth}{!}{%
\begin{tabular}{lcccccc}
\toprule
 & \multicolumn{4}{c}{Proposed Model} & \multicolumn{2}{c}{Benchmark} \\
\cmidrule(lr){2-5}\cmidrule(lr){6-7}
Time Point & MV & MV ($\rho=0$) & Midpoint & Midpoint ($\rho=0$) & DPM & ID \\
\midrule
Feb 11 (Mon) 12:00 & 13,472.4 & 13,304.8 & 13,455.9 & 13,282.0 & 18,609.3 & \textbf{10,981.1} \\
Feb 17 (Sun) 16:00 & 13,874.6 & 13,597.3 & 13,941.2 & 13,650.8 & 17,884.7 & \textbf{11,286.4} \\
Feb 20 (Wed) 14:00 & 24,897.1 & 24,543.9 & 24,957.2 & 24,585.7 & 33,833.0 & \textbf{20,038.8} \\
Jul 1 (Mon) 09:00 & 21,745.2 & 21,647.5 & 21,736.9 & 21,493.5 & 32,513.8 & \textbf{18,051.7} \\
Jul 20 (Sat) 16:00 & 15,396.6 & 15,157.0 & 15,574.0 & 15,335.8 & 20,943.8 & \textbf{12,594.2} \\
Jul 25 (Thu) 23:00 & 10,221.7 & 10,138.2 & 10,351.2 & 10,208.3 & 14,996.7 & \textbf{8,549.2} \\
Dec 11 (Wed) 05:00 & 9,276.9 & 9,155.3 & 9,429.7 & 9,242.1 & 13,449.6 & \textbf{7,663.8} \\
Dec 24 (Tue) 19:00 & 17,155.5 & 16,921.8 & 17,462.7 & 17,211.4 & 24,279.7 & \textbf{14,577.6} \\
Dec 31 (Tue) 18:00 & 9,049.8 & 8,957.9 & 9,090.4 & 8,983.0 & 12,469.7 & \textbf{7,644.9} \\
\bottomrule
\end{tabular}%
}
\label{tab:ppl-real-scale100-with-independent-dirichlet}
\end{table}

\begin{table}[htbp]
\centering
\scriptsize
\caption{Posterior predictive loss (PPL) on the nine real-data targets when the counts are reduced to one tenth of the original. Factor PT is shown with the MV and midpoint tree builders, each with either estimated $\rho$ or fixed $\rho=0$. 
}
\resizebox{\textwidth}{!}{%
\begin{tabular}{lcccccc}
\toprule
 & \multicolumn{4}{c}{Proposed Model} & \multicolumn{2}{c}{Benchmark} \\
\cmidrule(lr){2-5}\cmidrule(lr){6-7}
Time Point & MV & MV ($\rho=0$) & Midpoint & Midpoint ($\rho=0$) & DPM & ID \\
\midrule
Feb 11 (Mon) 12:00 & \textbf{857.0} & 892.7 & 858.8 & 893.2 & 1,168.7 & 1,091.3 \\
Feb 17 (Sun) 16:00 & \textbf{883.2} & 926.5 & 887.2 & 916.4 & 1,248.2 & 1,123.1 \\
Feb 20 (Wed) 14:00 & \textbf{1,646.5} & 1,708.4 & 1,655.5 & 1,723.1 & 2,430.0 & 1,997.8 \\
Jul 1 (Mon) 09:00 & \textbf{1,447.4} & 1,513.6 & 1,457.4 & 1,531.5 & 2,251.6 & 1,798.7 \\
Jul 20 (Sat) 16:00 & 1,001.4 & 1,026.0 & \textbf{999.1} & 1,029.7 & 1,361.8 & 1,254.1 \\
Jul 25 (Thu) 23:00 & \textbf{650.1} & 677.0 & 652.9 & 678.6 & 802.1 & 848.9 \\
Dec 11 (Wed) 05:00 & 574.8 & 599.7 & \textbf{573.9} & 599.2 & 694.8 & 759.9 \\
Dec 24 (Tue) 19:00 & \textbf{1,178.3} & 1,210.4 & 1,183.2 & 1,220.5 & 1,666.9 & 1,451.6 \\
Dec 31 (Tue) 18:00 & \textbf{578.9} & 595.5 & 585.1 & 598.7 & 715.8 & 757.6 \\
\bottomrule
\end{tabular}%
}
\label{tab:ppl-real-scale010-with-independent-dirichlet}
\end{table}

\subsection{Implication provided by the distributional factor analysis}
\label{sec: Application to The Spatiotemporal Population Data}

We next describe a detailed analysis of the results obtained with the proposed factor model, selecting the specific time and date at 19:00 on December 24th, as the implication is clear and highly interpretable.
Figure~\ref{fig: christmas tree} visualizes the result based on the tree constructed by the MV algorithm. 
The number of factors selected in the post-processing is $K^* = 4$ under the threshold 90\%.

For each factor index $k$ ($k=1,\dots,K^*$), one natural way to interpret the estimated factor is to compute the following Euclidean vectors
\[
    \{\mu(A) + c_k \eta_k(A)\}_{A \in \mathcal{N}(T)},
\]
where $c_k>0$ controls the magnitude of deviation along the $k$th factor direction, and we transform them into categorical distributions with the inverse embedding $\mathcal{E}_T^{-1}$.
In our analysis, we set $c_k$ to twice the standard deviation of the estimated factor loadings over the $M$ areas. 
If the $K^*$ factor distributions have similar structures with small visible differences, it is also effective to compute the difference between each factor distribution and the average observed distribution, which we calculate by the empirical mean of the observed proportions cell-wise.

The estimated factor distributions are shown in Figure~\ref{fig: estimated factor distributions (application)}.
It is seen that they exhibit various characteristics. 
For instance, compared to the average of the observed distributions, the first factor ($k=1$) exhibits lower probabilities especially for males in the age group 20--40, and higher probabilities for teenagers and females older than 60 years. 
Hence, a negative loading for the first factor places higher probabilities for males in their 20s to 40s, a demographic strongly associated with full-time employment. 
Conversely, a positive loading places higher probabilities for young people and older females. 
The results for the following factors ($k=2,3,4$) are also well interpretable because distributions with a positive loading for this factor increases the proportions of specific gender and/or age groups. 
The positive loadings are strongly associated with females in the age group 20--40 for $k=2$, young males and females aged 15--29 for $k=3$, and females in their 20s to 50s for $k=4$.

Figure~\ref{fig: estimated factor loadings (application)} visualizes the spatial distributions of the estimated factor loadings. 
These spatial distributions of the factor loadings, along with the estimated factor distributions, align well with socio-economic characteristics of Tokyo. 
For the first factor ($k=1$), the areas with negative loadings (indicated in blue) encompass major business districts and transportation hubs such as Tokyo and Shinagawa Stations (see Figure~\ref{fig:tokyomap} for the locations of these landmarks in Tokyo). 
Conversely, the areas with positive loadings (indicated in red) typically correspond to residential areas. 
The loadings for the second factor ($k=2$) are mostly positive across the study area, globally adjusting for female populations. 
For the third factor ($k=3$), the areas with positive loadings cover major commercial and entertainment areas, which attract young people. 
Lastly, the spatial distribution of the loadings for the fourth factor ($k=4$), which represents the stay population of young and middle-aged females, exhibits a more localized characteristic. 
The figure shows that the factor loading of Area~1, shown in Figure~\ref{fig:tokyomap}, resulted in a large positive value. 
This area is home to Tokyo Dome, a major multi-purpose stadium that functions as a central hub for large-scale sports and cultural events. 
In fact, on the specific date illustrated in the figure, a concert by a pop idol group popular among females was held there, as also observed in Figure~\ref{fig:tokyomap}. 
This result demonstrates that the proposed model can effectively capture such localized features in population distribution. 
The figure also shows positive loadings in a smaller neighborhood around Shibuya.

In summary, the results are highly interpretable; with the distributional analysis, we can estimate what demographic characteristics explain the variability in the population distributions in the region adaptively from the data, in the form of factor distributions.
Additionally, by mapping the estimated factor loadings, we can identify the areas in which the contribution of each factor is strong or weak. 
Such interpretable results are hard to obtain by estimating the distributions independently or with mixture models.

\begin{figure}[htbp]
    \centering
    \includegraphics[width=0.7\linewidth]{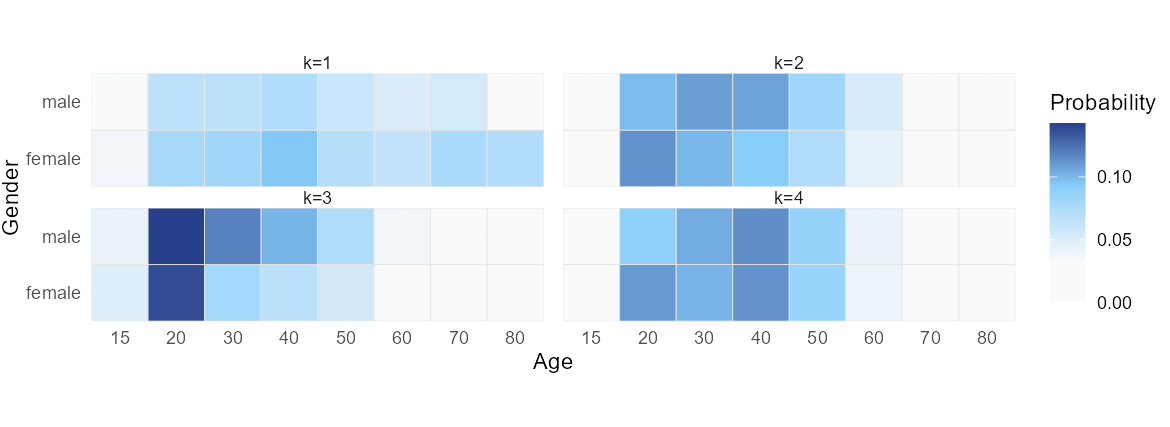}\\
    \includegraphics[width=0.7\linewidth]{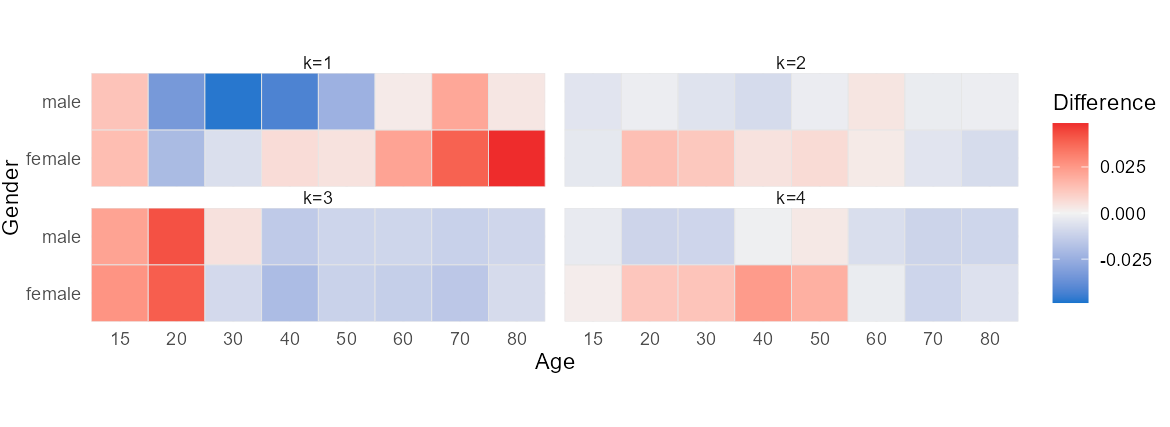}
    \caption{A visualization of the estimated factor distributions for the population data observed at 19:00 on 12/24/2019. 
    The factor distributions are displayed in  the top plot, and the difference from the average of observed distributions is shown in the bottom plot.
    }
    \label{fig: estimated factor distributions (application)}
\end{figure}

\begin{figure}[htbp]
    \centering
    \includegraphics[width=0.55\linewidth]{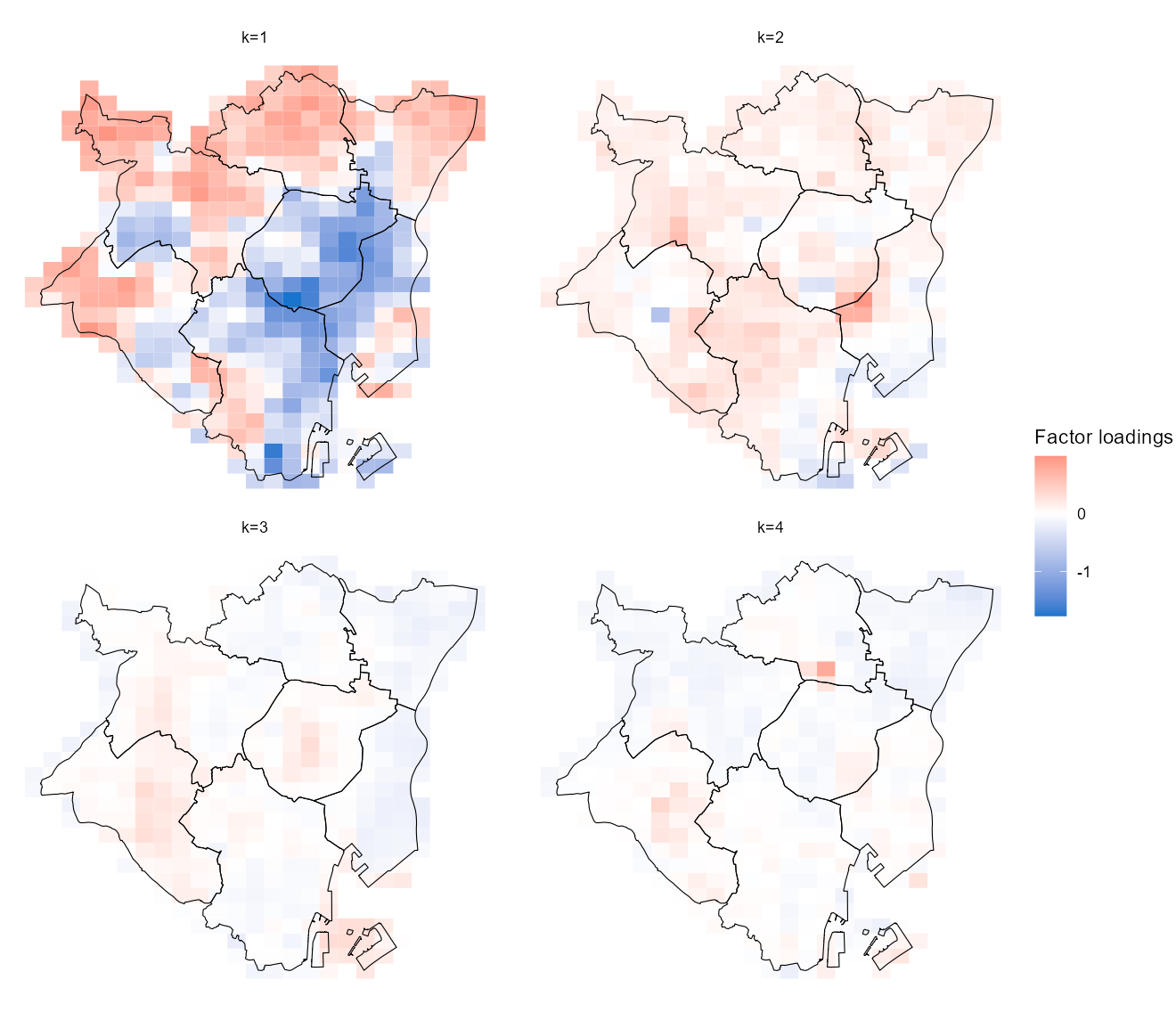}
    \caption{A visualization of the analysis of the population data (19:00 on 12/24/2019).
    The posterior means of the factor loadings defined for the $K^* = 4$ factors and the $M=452$ locations are visualized on the maps.
    }
    \label{fig: estimated factor loadings (application)}
\end{figure}

\section{Discussion}
\label{sec: discussion}
In this paper, we proposed a Bayesian distributional factor modeling framework that integrates tree-based representations of grouped data with infinite factor models to capture complex variability among multiple observed distributions through lower-dimensional latent factors. By embedding discrete distributions into Euclidean space via a logistic-tree transformation, the proposed approach enables flexible hierarchical modeling, as illustrated through the incorporation of a spatial SAR prior. The tree-based nonparametric representation allows diverse distributional shapes to be captured, while posterior computation remains tractable through \Polya Gamma augmentation. More broadly, the proposed framework provides a bridge between compositional data analysis and modern Bayesian factor modeling.

Through an empirical analysis of spatiotemporal population data in Tokyo, we demonstrated that the proposed method yields a parsimonious yet expressive characterization of heterogeneous demographic patterns, outperforming clustering-based and parametric alternatives in predictive performance under smooth spatial variations. These findings highlight the advantages of factor-based representations for distributional data exhibiting continuous variation rather than clearly separated cluster structures.

An important direction for future work is to extend the proposed framework to explicitly build a spatio-temporal dynamic model. For example, in the context of population data, such an extension could capture temporal patterns of distributions within a day or across seasons, and enable sequential prediction of counts for specific gender--age categories at each location.

\section{Acknowledgements}{
This work is partially supported by Japan Society for Promotion of Science (KAKENHI) grant number 24K00244, 25H00546, 25K16618, and 24H00142. 
}

\bibliographystyle{plainnat} 
\bibliography{references} 

\newpage

\appendix

\section{MCMC algorithm}
\label{appendix: MCMC algorithm}

The procedures to update each parameter in the MCMC algorithm are described below.
\begin{enumerate}
    \item (Updating $\omega_i(A)$) By \cite{polson2013bayesian}, the full conditional distribution of $\omega_i(A)$ is \[\mathrm{PG}(n_i(A), \psi_i(A)).\]
    \item (Updating $\bm{\Lambda}_{\cdot k}$)
    The full-conditional posterior of $\bm{\Lambda}_{\cdot k}$ is $N(\mathbf{m}_{\Lambda, k}, \bm{\Sigma}_{\Lambda, k})$, where
    \begin{align*}
        \bm{\Sigma}^{-1}_{\Lambda, k}
        &=
        \tau_k
        \left(
            I_M 
            -
            \rho_k W^T
        \right)
        \bm{\Phi}_k
        \left(
            I_M 
            -
            \rho_k W
        \right) 
        +
        \sum_{A \in \mathcal{N}(T)}
        \eta_k^2\, \bm{\Omega}(A), \\
        \bm{\Sigma}^{-1}_{\Lambda, k}\,
        \mathbf{m}_{\Lambda, k}
        &=
            \sum_{A \in \mathcal{N}(T)}
            \eta_k
            \left[
                \bm{\kappa}(A)
                -
                \bm{\Omega}(A)
                \left(
                \mu(A) \iota_M
                +
                \sum^K_{k'=1,\,k' \neq k}
                \eta_{k'}(A)
                \bm{\Lambda}_{\cdot k'}
                \right)
            \right].
    \end{align*}
    \item (Updating $\bm{\eta}(A)$)
    The full-conditional posterior of $\bm{\eta}(A)$ is $N(\mathbf{m}_{\bm{\eta}(A)}, \bm{\Sigma}_{\bm{\eta}(A)})$, where
    \begin{align*}
        \bm{\Sigma}_{\bm{\eta}(A)}^{-1}
        &=
        I_K
        +
        \bm{\Lambda}^T
        \bm{\Omega}(A)
        \bm{\Lambda}, \\
        \mathbf{m}_{\bm{\eta}(A)}
        &=
        \bm{\Sigma}_{\bm{\eta}(A)}
        \bm{\Lambda}^T
        \left[
            \bm{\kappa}(A) - 
            \mu(A) \bm{\Omega}(A) \iota_M
        \right].
    \end{align*}
    \item (Updating $\phi_{i,k}$)
    The full conditional posterior of $\phi_{i,k}$ is $\mathrm{Gamma}(a_{\phi, i,k}, b_{\phi, i, k})$, where
    \begin{align*}
        a_{\phi, i,k}
        =
        \frac{\nu + 1}{2},\
        b_{\phi, i, k}
        =
        \frac{1}{2}
        \left(
            \nu 
            +
            \tau_k
            \tilde{\bm{\Lambda}}_{i, k}^2
        \right), 
    \end{align*}
    where $\tilde{\bm{\Lambda}}_{\cdot k}
        =
        \{
            \tilde{\bm{\Lambda}}_{i, k}
        \}^M_{i=1}
        =
        (I_k - \rho_k W) \bm{\Lambda}_{\cdot k}$.
    \item (Updating $\tau_k$) We update $\tau_k$ via $\delta_1,\dots,\delta_K$ since $\tau_k = \prod^k_{k'=1} \delta_{k'}$. 
    To this end, we define $\tau^{(l)}_k$ for $l \geq k$ as
    \[
        \tau^{(l)}_k = \prod^l_{k'=1,\, k'\neq k} \delta_{k'}.
    \]
    Then, the full conditional posterior of $\delta_1$ is $\mathrm{Gamma}(a_{\delta, 1}, b_{\delta, 1})$, where
    \begin{align*}
        a_{\delta, 1}
        &=
        a_1 + \frac{MK}{2}, \\
        b_{\delta, 1}
        &=
        1 + 
        \frac{1}{2}
        \sum^{K}_{k'=1}
        \delta^{(1)}_{k'}
        \bm{\Lambda}_{\cdot k'}^T
        (I_M - \rho_{k'} W^T)
        \bm{\Phi}_{k'}
        (I_M - \rho_{k'} W)
        \bm{\Lambda}_{\cdot k'}.
    \end{align*}
    Similarly, the full conditional posterior for $\delta_k$ ($k=2,\dots,K$) is $\mathrm{Gamma}(a_{\delta, k}, b_{\delta, k})$, where
    \begin{align*}
        a_{\delta, k}
        &=
        a_2 + \frac{M(K-k+1)}{2}, \\
        b_{\delta, k}
        &=
        1 + 
        \frac{1}{2}
        \sum^{K}_{k'=k}
        \delta^{(k)}_{k'}
        \bm{\Lambda}_{\cdot k'}^T
        (I_M - \rho_{k'} W^T)
        \bm{\Phi}_{k'}
        (I_M - \rho_{k'} W)
        \bm{\Lambda}_{\cdot k'}.
    \end{align*}
    \item (Updating $\rho_k$) The full conditional posterior of $\rho_k$ is $\mathrm{Normal}_{(-\infty, \infty)}(m_{\rho, k}, s^2_{\rho, k})$, where
    \begin{align*}
        s^{-2}_{\rho, k}
        &=
        s^{-2}_\rho 
        +
        \tau_k \bm{\Lambda}_{\cdot k}^T W^T \bm{\Phi}_k W \bm{\Lambda}_{\cdot k},\\
        m_{\rho, k}
        &=
        s_{\rho,k}
        \left[
            s^{-2}_\rho m_\rho
            +
            \tau_k \bm{\Lambda}_{\cdot k}^T W^T \bm{\Phi}_k \bm{\Lambda}_{\cdot k}
        \right].
    \end{align*}
\end{enumerate}


\section{Proofs of theoretical properties}
\label{appendix:proofs}

\subsection{Independent loading case}
\label{appendix:proof_independent}
We first consider the setting without spatial correlation, where the factor loadings $\lambda_{ik}$ are drawn independently across locations.
Under the infinite factor model of \cite{bhattacharya2011sparse}, the latent parameter at location $i$ and internal node $A$ takes the form
\begin{equation}
  \psi_{i}(A)
  = \mu(A) + \sum_{k=1}^{\infty} \lambda_{ik}\,\eta_k(A),
  \label{eq:infinite_factor}
\end{equation}
where $\mu(A)$ is a fixed baseline, $\eta_k(A) \sim N(0,1)$ are independent standard normal latent factors, and the loadings follow the multiplicative gamma process prior with $\phi_{ik} \sim \mathrm{Ga}(\nu/2, \nu/2)$ ($\nu > 2$) and $\tau_k = \prod_{l=1}^k \delta_l$ with $\delta_1 \sim \mathrm{Ga}(a_1, 1)$ and $\delta_l \sim \mathrm{Ga}(a_2, 1)$ for $l \geq 2$ ($a_2 > 2$).
Let $P_i^{(0)}$ denote the true multinomial distribution at location $i$, with category probabilities $\{p_{i,l}^{(0)}\}_{l=1}^{N+1}$ where $N+1$ is the total number of categories.
The multinomial distribution $P_i^{(0)}$ is parameterized on the Pólya tree via the true conditional branching probabilities $\theta_i^{(0)}(A)$ at each internal node $A$, defined by the branching structure of the tree.
Applying the logit transformation to these branching probabilities yields the true Euclidean parameters $\psi_{i}^{(0)}(A) = \log\bigl(\theta_i^{(0)}(A) / (1 - \theta_i^{(0)}(A))\bigr) \in \mathbb{R}$.

In what follows, we assume that all categories in the true distribution have strictly positive probabilities, i.e., $p_{i,l}^{(0)} > 0$ for all $i$ and $l$.
We note that this assumption is made for clarity of exposition and can be relaxed: as discussed in Section~\ref{appendix:proof_boundary} below, posterior consistency continues to hold even when some categories have zero true probability, i.e., when the true distribution lies on the boundary of the probability simplex.

\begin{lem}
  \label{lem:truncation}
  Under the multiplicative gamma process prior with $\nu > 2$ and $a_2 > 2$:
  \begin{enumerate}
    \item[(i)] $P\!\left(\max_{1 \le i \le M}\sum_{k=1}^\infty \lambda_{ik}^2 < \infty\right) = 1$.
    \item[(ii)] For any $\xi > 0$, $\lim_{K \to \infty} P\!\left(\max_{i,A} |\psi_{i}(A) - \psi_{i,K}(A)| > \xi\right) = 0$.
  \end{enumerate}
\end{lem}
\begin{proof}
  Define $X_{ik} = \lambda_{ik}^2 \tau_k$.
  Since $\lambda_{ik} \mid \tau_k, \phi_{ik} \sim N(0, (\tau_k \phi_{ik})^{-1})$, we have $X_{ik} \mid \phi_{ik} \sim \phi_{ik}^{-1} \cdot \chi_1^2$.
  For $\gamma \in (0, 1/2)$, let $C_\gamma = E[|Z|^{2\gamma}] < \infty$ where $Z \sim N(0,1)$.
  Then $E[X_{ik}^\gamma] = E[ E[X_{ik}^\gamma \mid \phi_{ik}]] \le C_\gamma E[\phi_{ik}^{-\gamma}] \equiv C_\gamma M_\phi < \infty$, where $M_\phi < \infty$ follows from $\phi_{ik} \sim \mathrm{Ga}(\nu/2, \nu/2)$ with $\nu > 2$.
  The upper bound for the expectation $ C_\gamma M_\phi$ is denoted by $M_X$.

  By Markov's inequality, $P(X_{ik} > e^{\alpha k}) \le M_X e^{-\alpha \gamma k}$ for any $\alpha > 0$.
  Since $\sum_{k=1}^\infty e^{-\alpha\gamma k} < \infty$, the Borel--Cantelli lemma implies that $X_{ik} \le e^{\alpha k}$ for all sufficiently large $k$ with probability one; more precisely, there exists an almost sure event $\Omega_0$ such that on $\Omega_0$, there exists $K_0(\omega) < \infty$ with $X_{ik} \le e^{\alpha k}$ for all $k \ge K_0(\omega)$.
  Applying the strong law of large numbers to $\{\log\delta_l\}_{l \ge 2}$ (i.i.d.\ from $\mathrm{Ga}(a_2, 1)$) yields $k^{-1}\log\tau_k \to c_\delta := \psi(a_2) > 0$ almost surely, where $\psi(\cdot)$ is the digamma function and positivity follows from $a_2 > 2$.
  Therefore, on a single almost sure event, both $X_{ik} \le e^{\alpha k}$ and $\tau_k > e^{(c_\delta - \varepsilon)k}$ hold under $0 < \varepsilon < c_{\delta} $ for all sufficiently large $k$.
  Choosing $\alpha$ and $\varepsilon > 0$ such that $\beta := c_\delta - \varepsilon - \alpha > 0$, we obtain
  \[
    \lambda_{ik}^2 = \frac{X_{ik}}{\tau_k} \le e^{-(c_\delta - \varepsilon - \alpha)k} = e^{-\beta k}
    \quad \text{for all sufficiently large } k, \text{ almost surely.}
  \]
  Since $\sum_{k=1}^\infty e^{-\beta k} < \infty$ and the finitely many initial terms are finite almost surely, we conclude that $\sum_{k=1}^\infty \lambda_{ik}^2 < \infty$ almost surely for each fixed $i$. Taking a union bound over the finite set $\{1, \ldots, M\}$ of locations establishes part~(i).
  For part (ii), since $\sum_{k > K} \lambda_{ik}^2 \to 0$ almost surely as $K \to \infty$, the conditional Gaussian tail bound $P(|\psi_{i}(A) - \psi_{i,K}(A)| > \xi \mid \bm\Lambda) \le 2\exp(-\xi^2 / (2\sum_{k > K} \lambda_{ik}^2)) \to 0$ almost surely.
  Applying the dominated convergence theorem and a union bound over the finite index sets $i$ and $A$ completes the proof.
\end{proof}

\begin{cor}
    Under the multiplicative gamma process prior under $\nu > 2$ and $a_2 > 2$, 
    \[
        P
        \left(
            \max_{i,A}
            \left|
                \psi_{i}(A)
            \right|
            < \infty
        \right)
        = 1.
    \]
\end{cor}
\begin{proof}
Since $\psi_{i, K}(A)$ is a finite sum of finite random variables for each fixed $K < \infty$, 
\[
P\left(
\max_{i,A} |\psi_{i,K}(A)|<\infty
\right)=1 
\]
holds.
Therefore, 
for every fixed $K < \infty$ and every $\xi >0$, the following inclusion holds almost surely:
\[
\{
\max_{i,A}|\psi_{i}(A)|=\infty
\}
\subset
\{
\max_{i,A}|\psi_{i}(A)-\psi_{i,K}(A)|>\xi
\}.
\]
Consequently, we have
\[
P\left(
\max_{i,A}|\psi_{i}(A)|=\infty
\right)
\le
P\left(
\max_{i,A}|\psi_{i}(A)-\psi_{i,K}(A)|>\xi
\right).
\]
By \ref{lem:truncation} (ii), for every $\xi>0$,
\[
P\left(
\max_{i,A} |\psi_{i}(A)-\psi_{i,K}(A)|>\xi
\right) \to 0
\qquad \text{as } K\to\infty .
\]
Letting $K \to \infty$ in the preceding inequality, we obtain
\[
P\left(
\max_{i,A}|\psi_{i}(A)|=\infty
\right) =0,
\]
which proves the corollary.
\end{proof}

\begin{thm}
  \label{thm:support}
  For any $\{\psi_{i}^{(0)}(A)\}_{i,A} \subset \mathbb{R}^{M \times N}$ and any $\epsilon > 0$,
  \[
    P\!\left(\max_{i,A}\left|\psi_{i}(A) - \psi_{i}^{(0)}(A)\right| < \epsilon\right) > 0.
  \]
\end{thm}
\begin{proof}

Since the matrix of true parameters over the $M$ locations and $N$ internal nodes has rank at most $M$, choosing $K > M$ ensures that the true tail error is zero: $\sum_{k > K} \lambda_{ik}^{(0)} \eta_k^{(0)}(A) = 0$.

 Let us define the infinite series and its finite truncation up to $K$ as
$\psi_{i}^{(0)}(A) = \mu(A) + \sum_{k = 1}^{\infty} \lambda^{(0)}_{ik}\eta^{(0)}_k(A)$
and
$\psi_{i,K}^{(0)}(A) = \mu(A) + \sum_{k = 1}^{K} \lambda^{(0)}_{ik}\eta^{(0)}_k(A)$, respectively. 
By applying the triangle inequality, we can bound the maximum error as follows:
\[
  \max_{i,A} |\psi_{i}(A) - \psi_{i}^{(0)}(A)| 
  \le \underbrace{\max_{i,A} |\sum_{k > K} \lambda_{ik}\eta_k(A)|}_{=: E_1} + \underbrace{\max_{i,A} |\psi_{i,K}(A) - \psi_{i,K}^{(0)}(A)|}_{=: E_2}.
\]
  
Define $B_1 = \{E_1 < \epsilon/2\}$ and $B_2 = \{E_2 < \epsilon/2\}$ for any fixed $\epsilon > 0$.
  Given the global shrinkage parameters $\bm\delta = (\delta_1,\delta_2, \cdots)$, the events $B_1$ and $B_2$ are conditionally independent.

  By Lemma~\ref{lem:truncation}(ii), $P(B_1 \mid \bm\delta) \to 1$ almost surely as $K \to \infty$.
  For $B_2$, it is clear that $P(B_2 \mid \bm\delta) > 0$ almost surely. This follows because $\eta_k(A) \sim N(0,1)$ has a strictly positive density everywhere on $\mathbb{R}$, and the marginal density of $\lambda_{ik}$ given $\bm\delta$ (which is a scale mixture of normals obtained by integrating out $\phi_{ik}$) is likewise strictly positive on the entire real line $\mathbb{R}$.

  Taking expectations over $\bm\delta$ and applying Fatou's lemma then gives
  \[
    P\!\left(\max_{i,A}|\psi_{i} - \psi_{i}^{(0)}| < \epsilon\right)
    \ge E\!\left[P(B_1 \mid \bm\delta)\,P(B_2 \mid \bm\delta)\right] > 0,
  \]
  completing the proof.
\end{proof}

\begin{thm}
  \label{thm:consistency_independent}
  Under the infinite factor prior, the KL support condition holds: for any $\epsilon > 0$,
  \[
    P\!\left(\sum_{i=1}^{M} D_{KL}\!\left(P_i^{(0)}\,\|\,P_i\right) < \epsilon\right) > 0.
  \]
  Consequently, by Schwartz's theorem, the posterior distribution is consistent.
\end{thm}
\begin{proof}
  The KL divergence between $P_i^{(0)}$ and $P_i$ is defined as
  \[
    D_{KL}(P_i^{(0)} \| P_i) = \sum_{l=1}^{N+1} p_{i,l}^{(0)} \log \frac{p_{i,l}^{(0)}}{p_{i,l}}.
  \]
  Let $\theta(A)$ and $\theta^{(0)}(A)$ denote the modeled and true branching probabilities at internal node $A$, obtained by applying the logistic function to $\psi_{i}(A)$ and $\psi_{i}^{(0)}(A)$, respectively.
  Since the derivative of $\log \theta$ with respect to $\psi$ equals $1 - \theta \in (0,1)$, the mean value theorem gives
  \[
    |\log \theta(A) - \log \theta^{(0)}(A)| \le |\psi_{i}(A) - \psi_{i}^{(0)}(A)|,
  \]
  and the same bound holds for $|\log(1-\theta(A)) - \log(1-\theta^{(0)}(A))|$.
  Expanding $D_{KL}(P_i^{(0)} \| P_i)$ along the Polya-tree path representation \citep{awaya2024hidden} and applying this Lipschitz bound at each internal node, then summing over all paths and all locations, yields
  \[
    \sum_{i=1}^M D_{KL}(P_i^{(0)} \| P_i)
    \le C \max_{i,A}|\psi_{i}(A) - \psi_{i}^{(0)}(A)|,
  \]
  where $C =M \max_l |\mathcal{A}(l)|$ with $|\mathcal{A}(l)|$ denoting the depth of category $l$ in the tree.
  Setting $\delta = \epsilon/C$, the event $\{\max_{i,A}|\psi_{i}(A) - \psi_{i}^{(0)}(A)| < \delta\}$ implies $\{\sum_i D_{KL}(P_i^{(0)} \| P_i) < \epsilon\}$.
  By Theorem~\ref{thm:support}, the former event has strictly positive prior probability, which establishes the KL support condition, so the posterior consistency then follows from Schwartz's theorem \citep{ghosal2017fundamentals}.
\end{proof}

\subsection{SAR loading case}
\label{appendix:proof_spatial}

We now consider the spatial model in which the $k$th loading vector $\bm\Lambda_{\cdot k} = (\lambda_{1k}, \dots, \lambda_{Mk})^\top$ follows the SAR prior
\[
  (I_M - \rho_k W)\,\bm\Lambda_{\cdot k} \sim N\!\left(\bm{0},\,(\tau_k \bm\Phi_k)^{-1}\right),
\]
where $W$ is the row-normalized adjacency matrix, $\rho_k \sim \mathrm{TN}_{(-1,1)}(m_\rho, s_\rho^2)$, and $\bm\Phi_k$ and $\tau_k$ follow the same priors as in the independent case.

\begin{lem}
  \label{lem:spatial_sa}
  Under the SAR prior with $\rho_k \sim \mathrm{TN}_{(-1,1)}(m_\rho, s_\rho^2)$, $\phi_{ik} \sim \mathrm{Ga}(\nu/2, \nu/2)$ with $\nu > 2$, and the same multiplicative gamma process on $\tau_k$ with $a_2 > 2$,
  \[
    P\!\left(\max_{1 \le i \le M}\sum_{k=1}^\infty \lambda_{ik}^2 < \infty\right) = 1.
  \]
\end{lem}
\begin{proof}
  The Neumann series expansion of $(I_M - \rho_k W)^{-1}$ is valid since for the row sum norm $\|\rho_k W\|_\infty = |\rho_k|\|W\|_\infty = |\rho_k| < 1$, yielding $\|(I_M - \rho_k W)^{-1}\|_\infty \le (1 - |\rho_k|)^{-1}$.
  We define $\lambda_{ik}  \sim N(0, \sigma^2_{ik})$.
  Writing $\lambda_{ik} = \sigma_{ik} Z_{ik}$ with $Z_{ik} \sim N(0,1)$ and defining $X_{ik} = \lambda_{ik}^2 \tau_k$, the variance bound gives $X_{ik} \le (1 - |\rho_k|)^{-2}(\max_i \phi_{ik}^{-1}) Z_{ik}^2$.

  For $\gamma \in (0, 1/2)$, the conditional expectation satisfies
  \[
    E[X_{ik}^\gamma \mid \rho_k, \bm\Phi_k]
    \le C_\gamma (1 - |\rho_k|)^{-2\gamma} \bigl(\max_i \phi_{ik}^{-1}\bigr)^\gamma,
  \]
  where $C_\gamma$ is the finite constant defined in the proof of Lemma~\ref{lem:truncation}.
  Integrating over $\bm\Phi_k$ gives $E[(\max_i\phi_{ik}^{-1})^\gamma]\equiv \mathcal{M}_\phi < \infty$.
  For the expectation over $\rho_k$, letting $B_\rho = \sup_\rho f_{\rho_k}(\rho) < \infty$ be the upper bound of the truncated normal density, we obtain that
  \[
    E\!\left[(1 - |\rho_k|)^{-2\gamma}\right]
    \le 2B_\rho \int_0^1 (1 - \rho)^{-2\gamma}\,d\rho
    = \frac{2B_\rho}{1 - 2\gamma}
    \equiv M_\rho < \infty,
  \]
  where the integral converges because $1 - 2\gamma > 0$.
  Hence $E[X_{ik}^\gamma] \le C_\gamma \mathcal{M}_\phi M_\rho \equiv M_X < \infty$, and the rest of the proof (Markov's inequality, Borel--Cantelli, strong law of large numbers for $\tau_k$) follows the same steps as Lemma~\ref{lem:truncation}(i).
\end{proof}

The key point in the above proof is not the specific form of the SAR prior, but the finite fractional moment condition for $X_{ik}=\lambda_{ik}^2\tau_k$.
If $E(X_{ik}^{\gamma})<\infty$ for some $\gamma>0$, then the same argument as in the above proof yields the almost sure convergence of $\sum_{k=1}^\infty \lambda_{ik}^2$.
Thus, the same type of summability result can also be established for other spatial dependence structures, such as CAR priors or Gaussian process priors, whenever the corresponding prior specification satisfies this moment condition.
Once the summability result is obtained, the truncation-error argument applies in the same manner.

\begin{thm}
  \label{thm:consistency_spatial}
  Under the SAR factor model with the prior described above, for any $\epsilon > 0$,
  \[
    P\!\left(\sum_{i=1}^{M} D_{KL}\!\left(P_i^{(0)}\,\|\,P_i\right) < \epsilon\right) > 0.
  \]
  Hence the posterior distribution is consistent.
\end{thm}

\begin{proof}
  By Lemma~\ref{lem:spatial_sa}, $\sum_{k=1}^\infty \lambda_{ik}^2 < \infty$ almost surely.
  The truncation error therefore vanishes as in the independent case, and Lemma~\ref{lem:truncation}(ii) applies with $\sum_{k > K}\lambda_{ik}^2$.

  For the full-support condition on $\bm\Lambda_{\cdot k}$, since $|\rho_k| < 1$ all eigenvalues of $(I_M - \rho_k W)$ have strictly positive real parts, so the matrix is invertible.
  By Sylvester's law of inertia, $\Sigma_{\Lambda_k} = \tau_k^{-1}(I_M - \rho_k W)^{-1}\bm\Phi_k^{-1}(I_M - \rho_k W^\top)^{-1}$ is positive definite almost surely.
  Therefore, for each fixed truncation level $K$, conditional on the finite-dimensional hyperparameters $(\bm\delta_{1:K}, \bm\Phi_{1:K}, \bm\rho_{1:K})$, the joint prior density of $(\bm\Lambda_{\cdot 1}, \dots, \bm\Lambda_{\cdot K})$ is strictly positive on $(\mathbb{R}^M)^K$.
  Thus the same finite-dimensional support step used in the proof of Theorem~\ref{thm:support} remains valid in the spatial case, and together with the truncation property above this yields the KL support condition.
  The KL bound then follows from the same Lipschitz argument as in the proof of Theorem~\ref{thm:consistency_independent}.
\end{proof}

\subsection{Extension to zero-probability categories}
\label{appendix:proof_boundary}

\begin{proof}[Proof that posterior consistency holds when the true category probabilities include zeros]
Fix $\epsilon > 0$ arbitrarily.
  Suppose some categories have zero true probability.
  Let $D = \max_l |\mathcal{A}(l)|$ be the maximum tree depth.
  Choose $L > \log(2D / \epsilon)$ and define the truncated parameters
  \[
    \tilde\psi_{i}^{(0)}(A)
    = \begin{cases}
        L  & \text{if } \theta_i^{(0)}(A) = 1, \\
        -L & \text{if } \theta_i^{(0)}(A) = 0, \\
        \log\!\bigl(\theta_i^{(0)}(A) / (1 - \theta_i^{(0)}(A))\bigr) & \text{otherwise.}
      \end{cases}
  \]
  Let $\tilde P_i^{(0)}$ be the distribution induced by the truncated parameters.
  For any category $l$ with $p_{i,l}^{(0)} > 0$, the tree path to $l$ avoids branches with zero probability, so the truncated logit differs from the true logit only at nodes where $\theta_i^{(0)}(A) = 1$ for any $p_{i,l}^{(0)} > 0$.
  Using $\log(1 + e^{-L}) \le e^{-L}$, the approximation error satisfies $\log p_{i,l}^{(0)} - \log \tilde p_{i,l}^{(0)} \le D e^{-L}$, giving
  \[
    \sum_{i=1}^M D_{\mathrm{KL}}(P_i^{(0)} \| \tilde P_i^{(0)}) \le D e^{-L} < \frac{\epsilon}{2}.
  \]
  Since each $\tilde\psi_{i}^{(0)}(A) \in [-L, L]$ is finite, Theorem~\ref{thm:support} (applied to the truncated parameters as the ``true'' target) with $\xi = \epsilon/(2D)$ yields
  \[
    P\!\left(\max_{i,A}|\psi_{i}(A) - \tilde\psi_{i}^{(0)}(A)| < \xi\right) > 0.
  \]
   On this event, the same Lipschitz argument as in the proof of Theorem~\ref{thm:consistency_independent} gives
  \[
  \sum_{l=1}^{N+1} 
  p_{i,l}^{(0)}
  \log \frac{\tilde p_{i,l}^{(0)}}{p_{i,l}}
    \le
      \sum_{l=1}^{N+1} 
  \log \frac{\tilde p_{i,l}^{(0)}}{p_{i,l}}
    \le
    (\max_l |\mathcal{A}(l)|)\, \max_A |\psi_{i}(A) - \tilde\psi_{i}^{(0)}(A)|,
  \]
  and summing over $i$ yields
  \[
    \sum_{i=1}^M   
    \left[
    \sum_{l=1}^{N+1} 
  p_{i,l}^{(0)}
  \log \frac{\tilde p_{i,l}^{(0)}}{p_{i,l}}
  \right]
    \le  D \cdot \max_{i,A} |\psi_{i}(A) - \tilde\psi_{i}^{(0)}(A)|
    <  D \xi = \frac{\epsilon}{2}.
  \]
  Therefore, we obtain
  \[
    \sum_{i=1}^M D_{KL}(P_i^{(0)} \| P_i)
    \le   \sum_{i=1}^M D_{KL}(P_i^{(0)} \| \tilde P_i^{(0)}) +
    \sum_{i=1}^M   
    \left[
    \sum_{l=1}^{N+1} 
  p_{i,l}^{(0)}
  \log \frac{\tilde p_{i,l}^{(0)}}{p_{i,l}}
  \right]
    < \frac{\epsilon}{2} + \frac{\epsilon}{2} = \epsilon.
  \]
  This establishes the KL support condition.
\end{proof}

\section{Details of the numerical experiments}
\label{appendix: Details of the numerical experiments}

\subsection{Details of the simulation scenarios}
\label{subsec(appendix): Details of the simulation scenarios}

This appendix describes the data-generating mechanisms used in the simulation study.
For each location \(i=1,\ldots,M\), let
\[
s_i=(x_i,y_i)^\top \in [-1,1]^2
\]
denote its spatial coordinate. The coordinates are generated independently from the
uniform distribution on \([-1,1]^2\). At each location, the observed count vector is
generated as
\[
Y_i \mid N_i,p_i \sim \mathrm{Multinomial}(N_i,p_i),
\]
where
\[
N_i \sim \mathrm{Unif}
\left\{
\lfloor 0.5 n_{\mathrm{local}}\rfloor,
\ldots,
\lfloor 2 n_{\mathrm{local}}\rfloor
\right\}.
\]
The probability vector \(p_i\) is scenario-specific.

\subsubsection{One-Dimensional Scenarios}

In the one-dimensional scenarios, the support is an ordered eight-category partition
\[
B_1=[0,20),\quad
B_2=[20,30),\quad
B_3=[30,40),\quad
B_4=[40,50),
\]
\[
B_5=[50,60),\quad
B_6=[60,70),\quad
B_7=[70,80),\quad
B_8=[80,\infty).
\]
The location-specific probability vector is
\[
p_i=\{p_i(1),\ldots,p_i(8)\}.
\]

\begin{description}

\item[Two-parameter log-normal scenario.]
Let
\[
A_i \mid s_i \sim \mathrm{LogNormal}(\mu_i,\sigma_i^2),
\]
with
\[
\mu_i = 0.1 + x_i^2 + y_i^2,
\qquad
\log \sigma_i^2 = 0.2x_i + 0.2y_i.
\]
The grouped probabilities are
\[
p_i(j)
=
\Pr(A_i \in B_j \mid s_i)
=
\int_{B_j}
f_{\mathrm{LN}}(a;\mu_i,\sigma_i^2)\,da,
\qquad
j=1,\ldots,8.
\]
This scenario gives a smooth spatially varying ordered distribution whose local shape is
controlled by two log-normal parameters.

\item[Spatial mixture scenario.]
Let the latent continuous density at location \(i\) be
\[
f_i(a)
=
\sum_{k=1}^3
w_{ik}
f_{\mathrm{LN}}(a;\mu_k,\sigma_k^2),
\]
where
\[
(\mu_1,\mu_2,\mu_3)=(3.78,3.05,4.18),
\]
and
\[
(\log\sigma_1^2,\log\sigma_2^2,\log\sigma_3^2)
=
(\log 0.20,\log 0.10,\log 0.10).
\]
The spatially varying weights are
\[
w_{ik}
=
\frac{\exp(\eta_{ik})}
{\sum_{\ell=1}^3 \exp(\eta_{i\ell})},
\]
with
\[
\eta_{ik}
=
b_k
+
a_k
\exp\left(
-\frac{\lVert s_i-c_k\rVert^2}{\tau_k}
\right).
\]
The parameter values are
\[
b=(\log 8,0,0),
\]
\[
c_1=(0,0),\qquad
c_2=(-0.6,0.5),\qquad
c_3=(0.6,-0.4),
\]
\[
a=(0,\log 6,\log 6),
\qquad
\tau=(1.0,0.30,0.30).
\]
The grouped probabilities are
\[
p_i(j)
=
\int_{B_j} f_i(a)\,da,
\qquad
j=1,\ldots,8.
\]
This scenario creates spatially varying non-log-normal ordered distributions through
localized mixture components.

\end{description}

\subsubsection{Three-Dimensional Scenarios}

All three-dimensional scenarios use the same product support
\[
\mathcal X
=
\mathcal A \times \mathcal B \times \mathcal C,
\]
where
\[
\mathcal A=\{1,\ldots,8\},\qquad
\mathcal B=\{1,2\},\qquad
\mathcal C=\{1,2,3\}.
\]
For location \(i\), the joint probability is denoted by
\[
p_i(a,b,c),
\qquad
(a,b,c)\in \mathcal A\times\mathcal B\times\mathcal C.
\]
The first dimension \(\mathcal A\) uses the same ordered eight-category partition
\(\{B_1,\ldots,B_8\}\) as in the one-dimensional scenarios.

\begin{description}

\item[Axis-separated scenario.]
The joint distribution factorizes as
\[
p_i(a,b,c)
=
p_i^{(1)}(a)
p_i^{(2)}(b)
p_i^{(3)}(c).
\]
The first margin is generated from a nearly common grouped log-normal profile:
\[
A_i \mid s_i \sim
\mathrm{LogNormal}(3.78+\epsilon_i,0.20),
\qquad
\epsilon_i \stackrel{\mathrm{iid}}{\sim} N(0,0.03^2),
\]
and
\[
p_i^{(1)}(a)
=
\Pr(A_i\in B_a\mid s_i).
\]
The second margin varies along the \(x\)-axis:
\[
p_i^{(2)}(2)
=
0.20+0.60\,\mathrm{logit}^{-1}(4x_i),
\qquad
p_i^{(2)}(1)=1-p_i^{(2)}(2).
\]
The third margin varies along the \(y\)-axis. Let
\[
\omega_i=\mathrm{logit}^{-1}(4y_i).
\]
Then
\[
p_i^{(3)}
=
\omega_i(0.60,0.20,0.20)
+
(1-\omega_i)(0.20,0.20,0.60).
\]
This scenario separates the dominant spatial variation across different coordinate
directions.

\item[Latent-factor scenario.]
The joint distribution again factorizes as
\[
p_i(a,b,c)
=
p_i^{(1)}(a)
p_i^{(2)}(b)
p_i^{(3)}(c),
\]
but all three margins are driven by a shared latent spatial field. Let
\[
U=(U_1,\ldots,U_{M})^\top
\sim
N(m,\Sigma),
\]
where
\[
m_i=1.1-2.2(x_i^2+y_i^2),
\]
and
\[
\Sigma_{ii'}
=
0.60^2
\exp\left(
-\frac{\lVert s_i-s_{i'}\rVert^2}{2(0.50)^2}
\right)
+
10^{-8}\mathbf 1\{i=i'\}.
\]
The first margin is generated by a softmax model:
\[
\eta_{ia}^{(1)}
=
\alpha_a^{(1)}
+
\beta_a^{(1)}U_i
+
\epsilon_{ia}^{(1)},
\qquad
\epsilon_{ia}^{(1)}
\stackrel{\mathrm{iid}}{\sim}
N(0,0.08^2),
\]
with
\[
\alpha^{(1)}
=
(-0.25,0,0.35,0.45,0.20,-0.05,-0.25,-0.45),
\]
\[
\beta^{(1)}
=
(1.15,0.90,0.60,0.25,0,-0.30,-0.65,-0.95),
\]
and
\[
p_i^{(1)}(a)
=
\frac{\exp(\eta_{ia}^{(1)})}
{\sum_{a'=1}^8 \exp(\eta_{ia'}^{(1)})}.
\]
The second margin is
\[
p_i^{(2)}(2)
=
\mathrm{logit}^{-1}
\left(
0.32 U_i+\epsilon_i^{(2)}
\right),
\qquad
\epsilon_i^{(2)}
\stackrel{\mathrm{iid}}{\sim}
N(0,0.08^2),
\]
with \(p_i^{(2)}(1)=1-p_i^{(2)}(2)\).
The third margin is generated by another softmax model:
\[
\eta_{ic}^{(3)}
=
\alpha_c^{(3)}
+
\beta_c^{(3)}U_i
+
\epsilon_{ic}^{(3)},
\qquad
\epsilon_{ic}^{(3)}
\stackrel{\mathrm{iid}}{\sim}
N(0,0.08^2),
\]
where
\[
\alpha^{(3)}=(0.30,0,-0.30),
\qquad
\beta^{(3)}=(-0.75,0,0.75),
\]
and
\[
p_i^{(3)}(c)
=
\frac{\exp(\eta_{ic}^{(3)})}
{\sum_{c'=1}^3 \exp(\eta_{ic'}^{(3)})}.
\]
This scenario introduces a shared smooth latent spatial factor while keeping the three
margins conditionally independent given \(U_i\).

\item[Spatial-cluster scenario.]
The spatial domain is partitioned into a \(q\times q\) grid of rectangular regions.
In the paper-facing configuration, \(q=3\), giving \(9\) spatial clusters. Let
\[
c(i)\in\{1,\ldots,q^2\}
\]
denote the cluster membership of location \(i\). Define the baseline joint distribution
\[
\pi_0(a,b,c)
=
p_0^{(1)}(a)
p_0^{(2)}(b)
p_0^{(3)}(c),
\]
where
\[
p_0^{(1)}
=
(0.05,0.10,0.18,0.20,0.17,0.12,0.10,0.08),
\]
\[
p_0^{(2)}=(0.50,0.50),
\qquad
p_0^{(3)}=(0.35,0.40,0.25).
\]
For each cluster \(r=1,\ldots,q^2\), draw a full joint distribution
\[
\pi_r
\sim
\mathrm{Dirichlet}(25\pi_0).
\]
All locations in the same spatial cluster share the same true joint distribution:
\[
p_i(a,b,c)
=
\pi_{c(i)}(a,b,c).
\]
This scenario creates a piecewise-constant spatial truth and is intentionally favorable
to partition-based mixture models.

\end{description}

\subsection{Details of the models compared in the numerical experiments}
We describe the models compared to the proposed model in Section~\ref{sec: application to the real data}
\begin{enumerate}
    \item Independent Dirichlet model (ID)\\
    For each location, we independently fit a model with a Dirichlet prior
\[
    \mathrm{Dirichlet}(\alpha,\ldots,\alpha),
\]
where the hyperparameter $\alpha$ is set to $0.5$ in the experiments.
   \item Dirichlet process mixture model (DPM)
    \begin{align*}
    &\bm{Y}_i \mid \bm{\mu}_i  \sim \text{Multinomial}(N_i, \bm{\mu}_i), \quad i=1,\dots,M,\\
    &\bm{\mu}_i \mid G  \sim G , \quad i=1,\dots,M,\\
    & G \mid \alpha  \sim \text{DP}(\alpha, G_0) , \\
    & \alpha \sim \text{Gamma}(a, b), \\
    & G_0 = \text{Dirichlet}(\eta).
    \end{align*}
In this framework, $\bm{Y}_i$ denotes the vector of observed counts for the $i$-th observation, which follows a multinomial distribution with total count $N_i$ and a latent probability vector $\bm{\mu}_i$.
The latent vectors $\bm{\mu}_i$ are drawn from a random probability measure $G$, modeled as a Dirichlet process.
Here, the concentration parameter $\alpha$ is assigned a Gamma prior to control the clustering sparsity, while the base measure $G_0$ is specified as a symmetric Dirichlet distribution.
The default values of hyperparameters are $\eta=0.5$ and $a = b = 1$.
We adopt this model as a benchmark because the DPM is widely recognized as a standard Bayesian nonparametric approach for clustering, allowing for flexible modeling of heterogeneity without fixing the number of clusters a priori.
For posterior inference, we employ the Gibbs sampling algorithm described in \cite{neal2000DPM-alg} and \cite{escobar1995alpha-sample}.

\end{enumerate}

\section{Additional Numerical results}
\label{appendix: additional numerical results}
\subsection{Comparison of the proposed models with and without spatial correlation}
In this section, we compare the proposed factor model with and without the spatial correlation structure.
For the latter model, we fix the correlation parameter $\rho_k$ to zero, and additionally evaluate the model based on the KL divergence and the Hellinger distance, as in the numerical experiments (Section~\ref{sec: Numerical Experiments}).
The results are shown in  Table~\ref{tab:simulation-1d-factorpt-rho-comparison} and \ref{tab:simulation-3d-factorpt-rho-comparison}, and the results indicate that their performance is similar.

\begin{table}[htbp]
\centering
\scriptsize
\caption{Comparison of Factor PT models with estimated spatial correlation and fixed $\rho=0$ in the one-dimensional simulation scenarios. Entries show mean (SD) over 50 replicates. Reported values are $10^2 \times$ KL divergence and $10^2 \times$ Hellinger distance; lower values are better for both metrics.}
\resizebox{\textwidth}{!}{%
\begin{tabular}{llllcc}
\toprule
 & & & & \multicolumn{2}{c}{Metric} \\
\cmidrule(lr){5-6}
Scenario & Size & Model & Method & $10^2 \times$ KL div. & $10^2 \times$ Hellinger \\
\midrule
Two-param. & $n=1000$ & Factor PT & MV ($\rho$: estimated) & 0.061(0.006) & \textbf{1.075(0.055)} \\
 &  &  & MV ($\rho=0$) & 0.090(0.007) & 1.337(0.051) \\
 &  &  & Midpoint ($\rho$: estimated) & \textbf{0.061(0.006)} & 1.076(0.055) \\
 &  &  & Midpoint ($\rho=0$) & 0.089(0.006) & 1.332(0.048) \\
 & $n=200$ & Factor PT & MV ($\rho$: estimated) & 0.217(0.025) & 2.013(0.109) \\
 &  &  & MV ($\rho=0$) & 0.368(0.029) & 2.697(0.105) \\
 &  &  & Midpoint ($\rho$: estimated) & \textbf{0.216(0.025)} & \textbf{2.007(0.107)} \\
 &  &  & Midpoint ($\rho=0$) & 0.366(0.028) & 2.687(0.101) \\
\midrule
Mixture & $n=1000$ & Factor PT & MV ($\rho$: estimated) & \textbf{0.044(0.005)} & 0.931(0.056) \\
 &  &  & MV ($\rho=0$) & 0.073(0.007) & 1.204(0.055) \\
 &  &  & Midpoint ($\rho$: estimated) & \textbf{0.044(0.005)} & \textbf{0.924(0.059)} \\
 &  &  & Midpoint ($\rho=0$) & 0.072(0.007) & 1.200(0.056) \\
 & $n=200$ & Factor PT & MV ($\rho$: estimated) & 0.163(0.023) & 1.736(0.103) \\
 &  &  & MV ($\rho=0$) & 0.253(0.021) & 2.201(0.094) \\
 &  &  & Midpoint ($\rho$: estimated) & \textbf{0.161(0.024)} & \textbf{1.731(0.104)} \\
 &  &  & Midpoint ($\rho=0$) & 0.251(0.021) & 2.196(0.091) \\
\bottomrule
\end{tabular}%
}
\label{tab:simulation-1d-factorpt-rho-comparison}
\end{table}

\begin{table}[htbp]
\centering
\scriptsize
\caption{Comparison of Factor PT models with estimated spatial correlation and fixed $\rho=0$ in the three-dimensional simulation scenarios. Entries show mean (SD) over 50 replicates. Reported values are $10^2 \times$ KL divergence and $10^2 \times$ Hellinger distance; lower values are better for both metrics.}
\resizebox{\textwidth}{!}{%
\begin{tabular}{llllcc}
\toprule
 & & & & \multicolumn{2}{c}{Metric} \\
\cmidrule(lr){5-6}
Scenario & Size & Model & Method & $10^2 \times$ KL div. & $10^2 \times$ Hellinger \\
\midrule
3D axis-sep. & $n=5000$ & Factor PT & MV ($\rho$: estimated) & 0.038(0.002) & 0.937(0.030) \\
 &  &  & MV ($\rho=0$) & 0.046(0.002) & 1.024(0.025) \\
 &  &  & Midpoint ($\rho$: estimated) & \textbf{0.038(0.002)} & \textbf{0.936(0.028)} \\
 &  &  & Midpoint ($\rho=0$) & 0.045(0.002) & 1.016(0.020) \\
 & $n=1000$ & Factor PT & MV ($\rho$: estimated) & 0.155(0.010) & 1.882(0.057) \\
 &  &  & MV ($\rho=0$) & 0.207(0.011) & 2.184(0.059) \\
 &  &  & Midpoint ($\rho$: estimated) & \textbf{0.153(0.009)} & \textbf{1.875(0.054)} \\
 &  &  & Midpoint ($\rho=0$) & 0.193(0.011) & 2.101(0.057) \\
\midrule
3D latent & $n=5000$ & Factor PT & MV ($\rho$: estimated) & 0.114(0.009) & 1.642(0.071) \\
 &  &  & MV ($\rho=0$) & 0.107(0.009) & 1.597(0.072) \\
 &  &  & Midpoint ($\rho$: estimated) & 0.113(0.010) & 1.640(0.078) \\
 &  &  & Midpoint ($\rho=0$) & \textbf{0.107(0.009)} & \textbf{1.594(0.074)} \\
 & $n=1000$ & Factor PT & MV ($\rho$: estimated) & 0.534(0.055) & 3.596(0.205) \\
 &  &  & MV ($\rho=0$) & \textbf{0.432(0.056)} & \textbf{3.233(0.227)} \\
 &  &  & Midpoint ($\rho$: estimated) & 0.535(0.055) & 3.599(0.206) \\
 &  &  & Midpoint ($\rho=0$) & 0.437(0.058) & 3.253(0.235) \\
\midrule
3D cluster & $n=5000$ & Factor PT & MV ($\rho$: estimated) & \textbf{0.111(0.011)} & \textbf{1.907(0.113)} \\
 &  &  & MV ($\rho=0$) & 0.139(0.015) & 2.086(0.115) \\
 &  &  & Midpoint ($\rho$: estimated) & \textbf{0.111(0.011)} & \textbf{1.907(0.109)} \\
 &  &  & Midpoint ($\rho=0$) & 0.141(0.016) & 2.099(0.128) \\
 & $n=1000$ & Factor PT & MV ($\rho$: estimated) & \textbf{0.713(0.053)} & \textbf{4.991(0.222)} \\
 &  &  & MV ($\rho=0$) & 0.825(0.076) & 5.277(0.243) \\
 &  &  & Midpoint ($\rho$: estimated) & \textbf{0.713(0.057)} & 4.993(0.230) \\
 &  &  & Midpoint ($\rho=0$) & 0.831(0.096) & 5.294(0.291) \\
\bottomrule
\end{tabular}%
}
\label{tab:simulation-3d-factorpt-rho-comparison}
\end{table}

\subsection{Sensitivity Analysis on tree building algorithms}
\label{sec: Sensitivity Analysis on tree building algorithms}
In this analysis, we compare the following five algorithms to construct tree structures on the sample space $\Omega$.

\begin{description}
  \item[MV]
  At each node, this builder considers all admissible ordered binary splits along
  all category dimensions. The selected split is the one that maximizes the
  empirical variance, across locations, of the induced log-ratio balance. Thus,
  the tree is data-adaptive while preserving the ordering of categories.

  \item[Midpoint]
  This builder cycles through the category dimensions and splits the current
  range at its midpoint along the selected dimension. The resulting tree is
  roughly balanced in terms of the number of categories and does not use the
  observed category frequencies when choosing cut points.

  \item[Median]
  This builder also cycles through the category dimensions, but the cut point is
  chosen so that the average probability mass on the left child is as close as
  possible to one half. Ties are resolved toward the midpoint. The resulting
  tree preserves the category ordering while adapting the split locations to the
  empirical distribution.

  \item[Unbalanced]
  This builder cycles through the category dimensions and always separates the
  first remaining category level from the rest along the selected dimension.
  It therefore produces an ordered but intentionally unbalanced tree, useful as
  a contrast to more balanced constructions.

  \item[Random-midpoint]
  This builder first applies an independent random permutation to the category
  levels within each category dimension, and then applies the Midpoint rule to the permuted array. It preserves the
  product structure of the category dimensions, but removes the original
  ordering information. The random permutations are controlled by the tree seed.
\end{description}

We provide a comparison of these five algorithms, by estimating the simulation models considered in Section~\ref{sec: Numerical Experiments}.
The results based on the KL divergence, Hellinger distance, and the RMSE under two combinations of sample sizes are provided in Table~\ref{table: sensitivity (large sample)} and Table~\ref{table: sensitivity (small sample)}.
The result indicates that the tree building algorithm, and thus the fixed tree structures, do not severely influence the performance in terms of estimation accuracy. 

\begin{table}[htbp]
\centering
\scriptsize
\setlength{\tabcolsep}{3pt}
\caption{(\color{black}Table of sensitivity analysis \color{black})
Tree-sensitivity analysis for Factor PT under $K_{\max}=10$. Entries are reported as mean (SD) across 50 simulation replicates. }
\resizebox{\textwidth}{!}{%
\begin{tabular}{llcccc}
\toprule
 & & \multicolumn{4}{c}{Metric} \\
\cmidrule(lr){3-6}
Scenario & Tree Builder & PPL & KL div. & Hellinger & RMSE \\
\midrule
Two-param. ($n=1000$) & MV & 1,699(62) & 0.00061(0.00006) & 0.01075(0.00055) & 0.00473(0.00024) \\
 & Median & 1,696(57) & 0.00061(0.00006) & \textbf{0.01074(0.00053)} & 0.00471(0.00022) \\
 & Midpoint & \textbf{1,694(61)} & \textbf{0.00061(0.00006)} & 0.01076(0.00055) & \textbf{0.00470(0.00023)} \\
 & Unbalanced & 1,700(62) & 0.00062(0.00006) & 0.01074(0.00056) & 0.00471(0.00023) \\
 & Random-midpoint & 1,733(69) & 0.00072(0.00008) & 0.01173(0.00072) & 0.00501(0.00026) \\
\midrule
Mixture ($n=1000$) & MV & 2,171(64) & 0.00044(0.00005) & 0.00931(0.00056) & 0.00327(0.00020) \\
 & Median & 2,172(68) & 0.00044(0.00006) & 0.00931(0.00059) & 0.00328(0.00021) \\
 & Midpoint & \textbf{2,166(66)} & \textbf{0.00044(0.00005)} & \textbf{0.00924(0.00059)} & \textbf{0.00325(0.00021)} \\
 & Unbalanced & 2,170(66) & 0.00044(0.00005) & 0.00925(0.00058) & 0.00326(0.00021) \\
 & Random-midpoint & 2,174(69) & 0.00045(0.00006) & 0.00936(0.00063) & 0.00328(0.00023) \\
\midrule
3D axis-sep. ($n=5000$) & MV & 11,894(281) & 0.00038(0.00002) & 0.00937(0.00030) & 0.00061(0.00002) \\
 & Median & 11,894(277) & 0.00038(0.00002) & 0.00932(0.00025) & 0.00061(0.00002) \\
 & Midpoint & 11,892(279) & 0.00038(0.00002) & 0.00936(0.00028) & 0.00061(0.00002) \\
 & Unbalanced & 11,889(275) & 0.00038(0.00002) & 0.00927(0.00024) & 0.00061(0.00002) \\
 & Random-midpoint & \textbf{11,881(284)} & \textbf{0.00038(0.00003)} & \textbf{0.00925(0.00037)} & \textbf{0.00061(0.00002)} \\
\midrule
3D latent ($n=5000$) & MV & 12,211(302) & 0.00114(0.00009) & 0.01642(0.00071) & 0.00108(0.00005) \\
 & Median & 12,210(306) & 0.00113(0.00009) & 0.01636(0.00073) & 0.00108(0.00005) \\
 & Midpoint & 12,220(296) & 0.00113(0.00010) & 0.01640(0.00078) & 0.00108(0.00006) \\
 & Unbalanced & 12,178(298) & 0.00112(0.00010) & 0.01631(0.00077) & 0.00107(0.00005) \\
 & Random-midpoint & \textbf{12,163(306)} & \textbf{0.00111(0.00010)} & \textbf{0.01621(0.00080)} & \textbf{0.00106(0.00005)} \\
\midrule
3D cluster ($n=5000$) & MV & 11,477(281) & 0.00111(0.00011) & 0.01907(0.00113) & 0.00077(0.00005) \\
 & Median & 11,480(287) & 0.00111(0.00011) & 0.01911(0.00113) & 0.00074(0.00005) \\
 & Midpoint & 11,486(286) & 0.00111(0.00011) & 0.01907(0.00109) & \textbf{0.00074(0.00005)} \\
 & Unbalanced & \textbf{11,470(285)} & 0.00112(0.00011) & 0.01915(0.00113) & 0.00076(0.00005) \\
 & Random-midpoint & 11,474(280) & \textbf{0.00110(0.00011)} & \textbf{0.01904(0.00113)} & 0.00074(0.00005) \\
\bottomrule
\end{tabular}%
}
\label{table: sensitivity (large sample)}
\end{table}


\begin{table}[htbp]
\centering
\scriptsize
\setlength{\tabcolsep}{3pt}
\caption{Tree-sensitivity analysis for Factor PT under $K_{\max}=10$, \color{black}with DPM added as a benchmark.\color{black} The first five rows within each scenario correspond to the $\rho$-estimated Factor PT fits using different tree builders. Entries are reported as mean(SD) across 50 simulation replicates.}
\resizebox{\textwidth}{!}{%
\begin{tabular}{lllcccc}
\toprule
 & & & \multicolumn{4}{c}{Metric} \\
\cmidrule(lr){4-7}
Scenario & Model & Tree Builder  & PPL & KL div. & Hellinger & RMSE \\
\midrule
Two-param. ($n=1000$) & Factor PT (Proposed) & MV & 1,699(62) & 0.00061(0.00006) & 0.01075(0.00055) & 0.00473(0.00024) \\
 &  & Median & 1,696(57) & 0.00061(0.00006) & \textbf{0.01074(0.00053)} & 0.00471(0.00022) \\
 &  & Midpoint & \textbf{1,694(61)} & \textbf{0.00061(0.00006)} & 0.01076(0.00055) & \textbf{0.00470(0.00023)} \\
 &  & Unbalanced & 1,700(62) & 0.00062(0.00006) & 0.01074(0.00056) & 0.00471(0.00023) \\
 &  & Random-midpoint & 1,733(69) & 0.00072(0.00008) & 0.01173(0.00072) & 0.00501(0.00026) \\
 & DPM &  & 2,172(98) & 0.00163(0.00015) & 0.01810(0.00081) & 0.00790(0.00037) \\
\midrule
Mixture ($n=1000$) & Factor PT (Proposed) & MV & 2,171(64) & 0.00044(0.00005) & 0.00931(0.00056) & 0.00327(0.00020) \\
 &  & Median & 2,172(68) & 0.00044(0.00006) & 0.00931(0.00059) & 0.00328(0.00021) \\
 &  & Midpoint & \textbf{2,166(66)} & \textbf{0.00044(0.00005)} & \textbf{0.00924(0.00059)} & \textbf{0.00325(0.00021)} \\
 &  & Unbalanced & 2,170(66) & 0.00044(0.00005) & 0.00925(0.00058) & 0.00326(0.00021) \\
 &  & Random-midpoint & 2,174(69) & 0.00045(0.00006) & 0.00936(0.00063) & 0.00328(0.00023) \\
 & DPM &  & 2,303(74) & 0.00081(0.00011) & 0.01182(0.00080) & 0.00411(0.00028) \\
\midrule
3D axis-sep. ($n=5000$) & Factor PT (Proposed) & MV & 11,894(281) & 0.00038(0.00002) & 0.00937(0.00030) & 0.00061(0.00002) \\
 &  & Median & 11,894(277) & 0.00038(0.00002) & 0.00932(0.00025) & 0.00061(0.00002) \\
 &  & Midpoint & 11,892(279) & 0.00038(0.00002) & 0.00936(0.00028) & 0.00061(0.00002) \\
 &  & Unbalanced & 11,889(275) & 0.00038(0.00002) & 0.00927(0.00024) & 0.00061(0.00002) \\
 &  & Random-midpoint & \textbf{11,881(284)} & \textbf{0.00038(0.00003)} & \textbf{0.00925(0.00037)} & \textbf{0.00061(0.00002)} \\
 & DPM &  & 20,384(748) & 0.00462(0.00035) & 0.03201(0.00118) & 0.00214(0.00009) \\
\midrule
3D latent ($n=5000$) & Factor PT (Proposed) & MV & 12,211(302) & 0.00114(0.00009) & 0.01642(0.00071) & 0.00108(0.00005) \\
 &  & Median & 12,210(306) & 0.00113(0.00009) & 0.01636(0.00073) & 0.00108(0.00005) \\
 &  & Midpoint & 12,220(296) & 0.00113(0.00010) & 0.01640(0.00078) & 0.00108(0.00006) \\
 &  & Unbalanced & 12,178(298) & 0.00112(0.00010) & 0.01631(0.00077) & 0.00107(0.00005) \\
 &  & Random-midpoint & \textbf{12,163(306)} & \textbf{0.00111(0.00010)} & \textbf{0.01621(0.00080)} & \textbf{0.00106(0.00005)} \\
 & DPM &  & 26,340(1,894) & 0.00554(0.00028) & 0.03621(0.00090) & 0.00265(0.00013) \\
\midrule
3D cluster ($n=5000$) & Factor PT (Proposed) & MV & 11,477(281) & 0.00111(0.00011) & 0.01907(0.00113) & 0.00077(0.00005) \\
 &  & Median & 11,480(287) & 0.00111(0.00011) & 0.01911(0.00113) & 0.00074(0.00005) \\
 &  & Midpoint & 11,486(286) & 0.00111(0.00011) & 0.01907(0.00109) & 0.00074(0.00005) \\
 &  & Unbalanced & \textbf{11,470(285)} & 0.00112(0.00011) & 0.01915(0.00113) & 0.00076(0.00005) \\
 &  & Random-midpoint & 11,474(280) & 0.00110(0.00011) & 0.01904(0.00113) & 0.00074(0.00005) \\
 & DPM &  & 11,693(295) & \textbf{0.00017(0.00001)} & \textbf{0.00643(0.00025)} & \textbf{0.00037(0.00002)} \\
\bottomrule
\end{tabular}%
}
\label{table: sensitivity (small sample)}
\end{table}

\subsection{Additional results on the mixture models}
In this section, we report the average number of clusters estimated by the DPM model in the numerical experiments discussed in Section~\ref{sec: Numerical Experiments} and \ref{subsec: Evaluation with real Population data}, shown in Table~\ref{table: number of clusters (simulation)} and \ref{table: number of clusters (real data)}, respectively.
The results show that the numbers are moderate in the simulation settings but can become excessive in the real data analysis with the full sample.
A possible explanation for this result is that the distributions defined for the $M=452$ locations are diverse and do not have a clear clustering structure.

\begin{table}[htbp]
\centering
\caption{Posterior mean numbers of clusters estimated by DPM across the 50 simulation replicates. Entries are reported as mean (SD).}
\begin{tabular}{lcc}
\toprule
Scenario & Original data & Reduced to $1/5$ \\
\midrule
Two-param. & 18.7(1.3) & 8.6(0.8) \\
Mixture & 5.4(0.5) & 3.3(0.4) \\
3D axis-sep. & 25.8(2.1) & 10.4(1.0) \\
3D latent & 18.8(2.4) & 8.8(0.9) \\
3D cluster & 9.0(0.0) & 9.0(0.0) \\
\bottomrule
\end{tabular}
\label{table: number of clusters (simulation)}
\end{table}

\begin{table}[htbp]
\centering
\caption{Posterior mean numbers of clusters estimated by DPM on the nine real-data targets.}
\begin{tabular}{lcc}
\toprule
Time Point & Original data & Reduced to $1/10$ \\
\midrule
Feb 11 (Mon) 12:00 & 106.1 & 10.9 \\
Feb 17 (Sun) 16:00 & 114.6 & 12.0 \\
Feb 20 (Wed) 14:00 & 164.0 & 20.0 \\
Jul 1 (Mon) 09:00 & 137.9 & 11.0 \\
Jul 20 (Sat) 16:00 & 119.6 & 16.0 \\
Jul 25 (Thu) 23:00 & 85.8 & 9.0 \\
Dec 11 (Wed) 05:00 & 87.7 & 6.0 \\
Dec 24 (Tue) 19:00 & 129.2 & 15.7 \\
Dec 31 (Tue) 18:00 & 65.8 & 9.0 \\
\bottomrule
\end{tabular}
\label{table: number of clusters (real data)}
\end{table}

\end{document}